\pdfoutput=1
\documentclass[reqno]{amsproc}
\usepackage{amssymb}
\usepackage{amsthm}
\usepackage{physics}
\usepackage{mathtools}
\usepackage{hyperref}
\usepackage{microtype}
\usepackage{tikz}
\usetikzlibrary{arrows.meta,decorations.pathreplacing}
\tikzset{> = {Straight Barb[scale=.75]}}

\usepackage[citation-order,nobysame]{amsrefs}
\usepackage{xyzbib}

 \mathtoolsset{showonlyrefs}

\numberwithin{equation}{section}
\let\cite=\cites

\newtheorem{theorem}{Theorem}
\newtheorem{proposition}[theorem]{Proposition}

\newtheorem*{remark}{Remark}

\newcommand{\FPQ}[3]{ 
	\,{}_{3}F_{2}
	\left( 
		\left. \genfrac{}{}{0pt}{}{ #1 }{ #2 } \right| #3
	\right)
}

\newcommand{\xct}{\tilde{x}_{\mathrm{c}}}
\newcommand{\xc}{x_{\mathrm{c}}}
  
\begin{document}

\vspace*{2cm}

\title[Arctic curves of the four-vertex model]{Arctic curves of the
  four-vertex model}

\author{I. N. Burenev} \address{Steklov Mathematical Institute,
  Fontanka 27, 191023 Saint Petersburg, Russia}
\email{inburenev@gmail.com}

\author{F. Colomo} \address{INFN, Sezione di Firenze  Via G. Sansone
  1, I-50019 Sesto Fiorentino (FI), Italy} \email{colomo@fi.infn.it}

\author{A. Maroncelli} \address{Dipartimento di Fisica e Astronomia,
  Universit\`a di Firenze,  and INFN, Sezione di Firenze, Via G. Sansone
  1, I-50019 Sesto Fiorentino (FI), Italy}
\email{andrea.maroncelli@unifi.it}

\author{A. G. Pronko}
\address{Steklov Mathematical Institute, 
  Fontanka 27, 191023 Saint Petersburg, Russia,  and
  Theoretical Physics Department, Saint Petersburg State University,
Ulyanovskaya str. 1, Peterhof, Saint Petersburg, 198504, Russia}
\email{agp@pdmi.ras.ru}

\begin{abstract}
We consider the four-vertex model with a special choice of fixed
boundary conditions giving rise to limit shape phenomena.  More
generally, the considered boundary conditions relate vertex models to
scalar products of off-shell Bethe states, boxed plane partitions, and
fishnet diagrams in quantum field theory.  In the scaling limit, the
model exhibits the emergence of an arctic curve separating a central
disordered region from six frozen `corners' of ferroelectric or
anti-ferroelectric type.  We determine the analytic expression of the
interface by means of the Tangent Method.  We supplement this
heuristic method with an alternative, rigorous derivation of the
arctic curve. This is based on the exact evaluation of suitable
correlation functions, devised to detect spatial transition from order
to disorder, in terms of the partition function of some discrete
log-gas associated to the orthogonalizing measure of the Hahn
polynomials.  As a by-product, we also deduce that the arctic curve's
fluctuations are governed by the Tracy-Widom distribution.

\end{abstract}

\maketitle

\tableofcontents


\section{Introduction}

Limit shape phenomena have been addressed for the first time in the
context of representation theory of the symmetric group
\cite{KV-77,LS-77}. Subsequently, it has been recognized independently
that these phenomena can be observed in random walks, random tilings,
and more generally, in many problems which admit statistical mechanics
formulation \cite{F-84,EKLP-92,KZj-00}. Much information, mostly
qualitative, can be extracted by means of numerical simulations; this
has stimulated the development of various approaches based on
Monte-Carlo methods and Markov processes
\cite{PW-96,SZ-04,E-99,AR-05}.

More interesting, however, is to gain an exact analytical description
of these phenomena.  For those tiling problems which can be seen as
dimer models, such as, domino or lozenge tilings, various methods
based on the use of Gelfand triangles or non-intersecting lattice
paths can be used to obtain exact analytical expressions for systems
of finite size. Having these expressions, one can then analize their
behaviour in the thermodynamic (scaling) limit by means of, e.g., the
saddle-point method. In this way the Arctic circle theorems have been
established for domino tiling of large Aztec diamonds \cite{CEP-96}
and lozenge tilings of large hexagons \cite{CLP-98,BGR-10}.

In a rather general setup, the problem can be formulated as a
variational principle of minimization of certain functionals
\cite{D-98,CKP-01,Zj-02,PR-08}. The solution of the variational
principle essentially describes the limit shape (the arctic curve
being the boundary of its flat facets). The problem has been solved in
full generality for dimer models on periodic planar bipartite graphs
\cite{KOS-06,KO-05}. For a pedagogical and comprehensive treatment of
the case of the honeycomb lattice, see \cite{G-21}. For recent
progresses in relation to nonperiodic graphs, see \cite{BCdT-23}, and
reference therein. Concerning models that cannot be expressed in terms
of free fermions, the variational principle has been solved only for
the five-vertex model \cite{KP-21,KP-22,dGKW-21}, and for some
particular cases of the stochastic six-vertex model
\cite{BCG-16,RS-18}.

When dealing with more complicated models, for which no method is
currently available for the determination of the limit shape, a less
ambitious but still interesting goal is that of evaluating the arctic
curve. To this aim we may resort to the Tangent Method
\cite{CS-16}. The method has its origin in the study of the arctic
curve of the six-vertex model with domain wall boundary conditions
\cite{CP-09,CP-08,CPZj-10}.  Although somewhat heuristic, the method
is remarkably simple, especially in comparison to the variational
approach.  It is applicable to a wide class of models, provided that a
lattice path formulation is available, and that a suitable boundary
one-point correlation function may be computed. Recent advances are
related to the construction of a rigorous proof of the method
\cite{A-20,DGR-19} and to its extension to various domains
\cite{CPS-18,DG-19,DR-21} and models \cite{DDG-20,D-21,D-23}.

In the present paper, we consider the four-vertex model \cite{LPW-90}
with a particular choice of fixed, `scalar-product', boundary
conditions, so called because the corresponding partition function can
be expressed as a scalar product of off-shell Bethe states, see
\cite{B-08,BuP-21} for details. These are sometimes
called `boxed plane partitions' boundary conditions \cite{dGKW-21}.
They also appear in the context of fishnet diagrams in $\mathcal{N}=4$
supersymmetric Yang-Mills theory \cite{OV-22}.

We address the problem of determination of the curves separating the
various phases. Numerical simulations of the model clearly show the
existence of three phases: ferroelectric order, disorder, and
anti-ferroelectric order \cite{BoP-21}. We apply the Tangent Method to
derive the analytic expression of the various portions of the arctic
curve. In particular, we show how the Tangent Method can be adapted to
obtain those portions which separate anti-ferroelectric order from
disorder.

We supplement the above heuristic approach with an alternative,
rigorous derivation of the artic curve. This is based on the
evaluation of some suitable correlation functions devised to detect
spatial transition from order to disorder.  Specifically we consider
the Emptiness Formation Probability (EFP) to study the transition from
ferroelectric order to disorder \cite{CP-07b}, and introduce a similar
correlation function, the Anti-ferroelectric phase Formation
Probability (AFP) to tackle the transition from anti-ferroelectric
order to disorder. We shall refer generically to this approach as the
`EFP Method' \cite{CP-07a}.

In practice, by mapping the four-vertex model to the five-vertex model
at its free-fermion point, we manage to evaluate explicitly EFP and
AFP. In both cases, the obtained expressions are easily recognized as
the gap probabilities for some discrete log-gases associated to the
orthogonalizing measure of the Hahn polynomials. The study of the
asymptotic behaviour of the log-gas in the scaling limit allows to
work out the explicit expression of the arctic curve. The log-gas
description also implies that the local fluctuations of the arctic
curve are described by the Tracy-Widom
distribution \cite{TW-94a,TW-94b}.

\subsection{The model}\label{sec:model}
The four-vertex model is a special case of the six-vertex model
\cite{B-82}, in which two vertex configurations are deleted, that is,
their Boltzmann weights are set equal to zero \cite{LPW-90}. The
remaining vertex configurations, with their Boltzmann weights $a$,
$b$, $c$, are given in Fig.~\ref{fig:weights}, where we have chosen a
representation in terms of thick/thin, or full/empty edges. The
four-vertex model may thus be viewed in terms of non-intersecting
(oriented) lattice paths on a square lattice, with an additional
constraint forbidding two consecutive steps in the horizontal
direction.

\begin{figure}[t]
\includegraphics{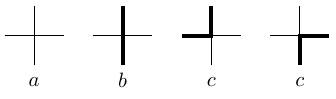}
\caption{The four vertex configurations of the model, with their
  Boltzmann weights.}\label{fig:weights}
\end{figure}

In the present paper we consider the four-vertex model on a
rectangular domain of size $L \times M$, with lattice coordinates
$(n,m)\in[1,L]\times[1,M]$, oriented in the standard way, that is
rightward and upward, respectively. The `scalar-product' boundary
conditions are fixed as follows: first (last) $N$ vertical edges of
the south (north) boundary are thick, with all other boundary edges
being thin, see Fig.~\ref{fig:bc}. The condition $M\geq L >N\geq 1$
guarantees that the model has more than just one admissible
configuration.

\begin{figure}[t]
\includegraphics{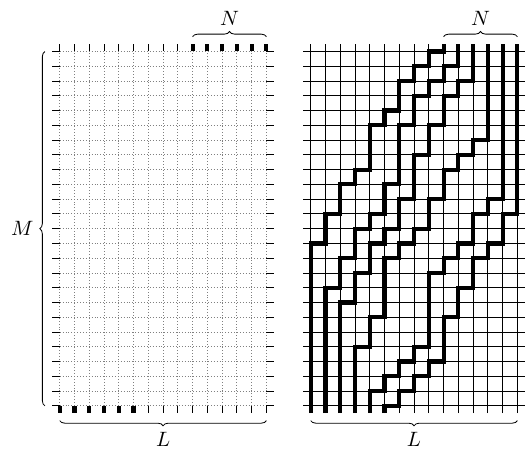}
\caption{Boundary conditions of the model (left) and a possible 
  configuration (right). Here $L=15$, $M=25$, and $N=6$.}
\label{fig:bc}
\end{figure} 

The partition function of the model is defined as
\begin{equation}
	Z_{L,M,N} (a,b,c) = \sum_{\text{conf}} a^{\# a} b^{\# b} c^{\# c}
\end{equation}
where the sum is taken over all possible configurations of the model
and $\#a$, $\#b$ and $\#c$ denote the number of vertices of type $a$,
$b$ and $c$, respectively.

A specific feature of the four-vertex model with fixed boundary
conditions is that the number of vertices of each type does not depend
on the configuration. In particular, in our case we have
\begin{equation}
	\# a = (L-N) (M-N), \quad
	\# b = N (M-L+N), \quad 	
	\# c = 2 N(L-N),
\end{equation} 
and thus 
\begin{equation}
	Z_{L,M,N}(a,b,c) = a^{(L-N)(M-N)} b^{N(M-L+N)} c^{2N(L-N)} Z_{L,M,N}
\end{equation}
where $Z_{L,M,N}$ is the number of allowed configurations. Therefore,
with no loss of generality, we set $a=b=c=1$. Correspondingly, the
Gibbs measure is uniform on the space of configurations of the model,
and equal to $1/Z_{L,M,N}$.

\subsection{Statement of the problem}\label{sec:numerics}

The four-vertex model may exhibit, under suitable choice of fixed
boundary conditions, the limit shape phenomenon.  In particular, in
the case of the scalar-product boundary conditions, in the scaling
limit one observes the emergence of a central disordered region and
six frozen regions: four of ferroelectric order (only $a$-, or only
$b$-vertices) and two of anti-ferroelectric order (only $c$-vertices,
alternating and forming a zig-zag pattern), see
Fig.~\ref{fig:conf_typical}.

\begin{figure}[t]
\includegraphics[width=.3\linewidth]{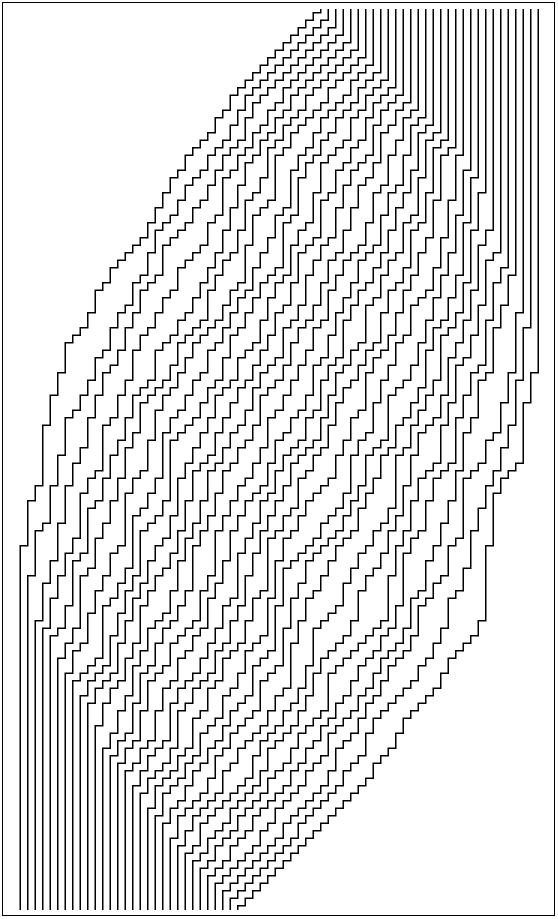}
\caption{A typical configuration of the model. The central disordered
  region and the six ordered ones, of type $a$, $b$, and $c$, are
  clearly visible. Here $N=30$, $L=70$, and
  $M=120$.}\label{fig:conf_typical}
\end{figure}

To further illustrate this phenomenon, we have performed the following
numerical experiment. We have generated $10^5$ configurations of the
model on the lattice with $N=60$, $L=140$, $M=240$. We have used the
Markov Chain Monte Carlo method with Coupling From The Past
\cite{PW-96}, to ensure uniform sampling.  In Fig.~\ref{fig:numerics}
we show the density of vertices of type $a$, $b$, and $c$,
respectively, averaged over $10^5$ configurations.

\begin{figure}[t]
\includegraphics[width=\linewidth]{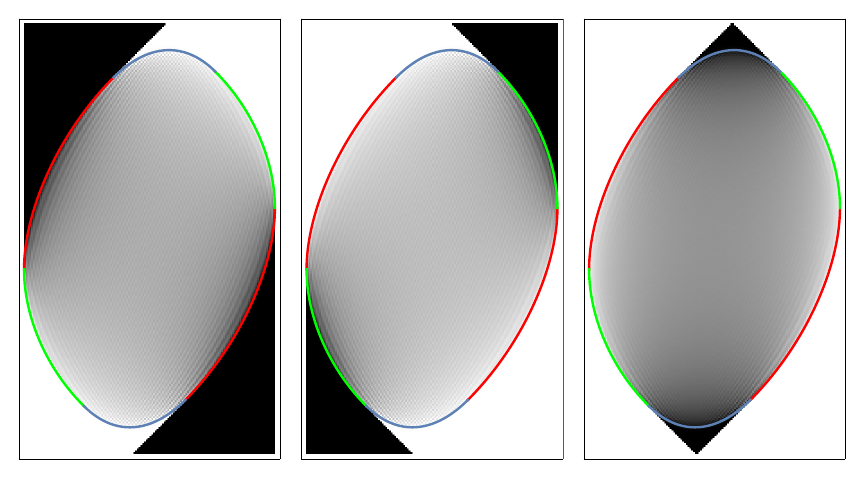}
\caption{Densities of vertices of type $a$, $b$, and $c$ (left to
  right), with color varying from white to black as densities vary
  from 0 to 1. The picture has been generated by averaging over $10^5$
  configurations of a lattice with $N=60$, $L=140$, $M=240$. The
  arctic curve, as given in Thm.~1, below, is also
  plotted.}\label{fig:numerics}
\end{figure}

The simulation clearly shows the emergence of the six frozen regions of
type $a$, $b$, and $c$, in black in the first, second, and third
picture, respectively. These are sharply separated from the central
disordered region by a smoth curve, known as the arctic curve.

The purpose of the present paper is twofold: on the one hand, to
provide an explicit analytic expression for the arctic curve of the
model, and on the other hand, to test and improve available methods to
make them more effective, and applicable to a broader range of
situations. In particular, below, we extend the Tangent Method to make
it applicable in situations where the arctic curve separates the
disordered region from an anti-ferroelectically frozen one, such as in
the top part of Fig.~\ref{fig:numerics}, right.  Also, we test the
still conjectural validity of the Tangent Method by providing an
alternative, rigorous derivation of our results.

\subsection{Main result}
Our main result is an explicit expression for the arctic curve of the
model. This emerges in the scaling limit, defined as follows. First
the sizes of the considered domain $L$, $M$, and $N$, and the lattice
coordinates $n$ and $m$, are rescaled by a parameter $\ell$, 
\begin{equation}\label{eq:scaling}
  L=\lceil\mathcal{L}\ell\rceil,\quad
     M=\lceil\mathcal{M}\ell\rceil,\quad
     N=\lceil\mathcal{N}\ell\rceil,\quad
     n=\lceil x\ell\rceil,\quad
     m=\lceil y\ell\rceil,
\end{equation}
and next the limit $\ell\to\infty$ is taken. The coordinates $x$ and
$y$ parametrize the rescaled domain, $(x,y)\in [0,\mathcal{L}]\times
[0,\mathcal{M}]$.

It appears that the arctic curve is made of six consecutive arcs,
joined end by end at six `contact points', $P_1,\dots,P_6$, see
Fig.~\ref{fig:AC}. Each arc is a specific portion of some algebraic
curve (actually, of some ellipse).  We denote these six arcs by
$\Gamma_j$, $j=1,\dots,6$, starting from the west contact point,
$P_1$, and in clockwise order. These arcs separate a central
disordered region from six different frozen regions of type $a$, $c$,
$b$, and again $a$, $c$, $b$, respectively.

\begin{figure}[t]
\includegraphics{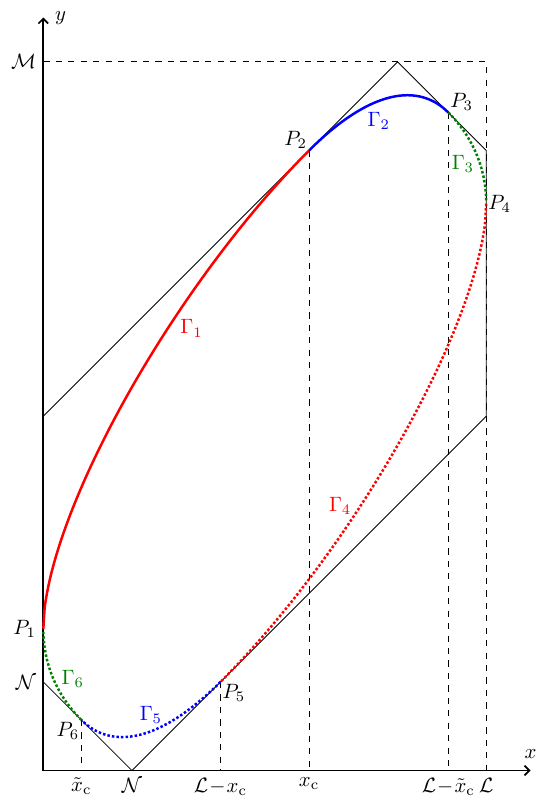}
\caption{Arctic curve for the four-vertex model, with its six contact
  points.  The solid portion is given by Thm.~\ref{thm:result}. The
  dotted portion may be restored from the symmetries of the
  model. Here $\mathcal{N:L:M}=1:5:8$.}\label{fig:AC}
\end{figure}

\begin{theorem}\label{thm:result} 
The portions $\Gamma_1$ and $\Gamma_2$ of the arctic curve of the
four-vertex model with the scalar-product boundary conditions are
described by the following equations:
\begin{equation}
  \left\{\begin{aligned}
  \Gamma_1:\quad & y=f_1(\mathcal{L},\mathcal{M},\mathcal{N};x), \quad
  && x\in(0,\xc] ,\\
  \Gamma_2:\quad&y=f_2(\mathcal{L},\mathcal{M},\mathcal{N};x),\quad
			&& x\in[\xc, \mathcal{L}-\xct],
	\end{aligned}
	\right.
\end{equation}
where
\begin{align}
	f_1(\mathcal{L},\mathcal{M},\mathcal{N};x) &= 
	\frac{	\mathcal{MN} (\mathcal{L}-2x)+ (\mathcal{M+N})\mathcal{L}x
        }{ \mathcal{L}^2 } 
        \\
        &\qquad \qquad
	 +2 \frac{ \sqrt{ \mathcal{MN(L-N)(M-L)} (\mathcal{L}-x)x } 
	   }
         {\mathcal{L}^2},\label{eq:ac-answer1}
         \\
         f_2(\mathcal{L},\mathcal{M},\mathcal{N};x) &= 
         (\mathcal{L-M-N}-x) + 2f_1(\mathcal{L},\mathcal{M},\mathcal{N};x),
         \label{eq:ac-answer2}
\end{align}
and
\begin{equation}\label{eq:xc}
\xc = \frac{(\mathcal{M}-\mathcal{L})(\mathcal{L} - \mathcal{N})}
{\mathcal{M}-\mathcal{L}+\mathcal{N}},\qquad
\xct=\frac{(\mathcal{M}-\mathcal{L})\mathcal{N}}{\mathcal{M-N}}.
\end{equation}
\end{theorem}
\begin{remark}
Below, we provide two different derivations of this result. The first
one is based on the Tangent Method \cite{CS-16}, and its validity is
therefore conditioned to that of the main assumptions underlying the
method, see Sect.~\ref{sec:TMmethod} for details. The second
derivation, based on the study of some suitable correlation function,
is instead fully rigorous, see Sect.~\ref{sec:EFPmethod}.
\end{remark}

One can easily find the expression for the remaining portions of the
curve with the help of symmetry arguments. In particular, $\Gamma_3$
may be obtained from $\Gamma_1$ by means of the so-called
`particle-hole' symmetry, which implies, in the scaling limit, the
invariance of the model under the simultaneous replacement
$\mathcal{N}\to\mathcal{L-N}$, $x\to\mathcal{L}-x$, see
Sec.~\ref{sec:symmetries} for details. We have
\begin{equation}\label{eq:gamma3}
  \Gamma_3:\quad  y=
    f_1(\mathcal{L},\mathcal{M},\mathcal{L-N};\mathcal{L}-x),\qquad 
    x\in[\mathcal{L}-\xct,\mathcal{L}].
\end{equation}
Note in particular that $\xct$ in \eqref{eq:xc} is obtained from $\xc$
by replacing $\mathcal{N}\to\mathcal{L-N}$.  Also, the particle-hole
symmetry, when applied to $\Gamma_2$, implies that
\begin{equation}
  f_2(\mathcal{L},\mathcal{M},\mathcal{L-N};\mathcal{L}-x)
  =f_2(\mathcal{L},\mathcal{M},\mathcal{N};x),\qquad
  x\in[\xc,\mathcal{L}-\xct],
\end{equation}
as it may be easily verified.  Concerning the remaining three portions
of the arctic curve, $\Gamma_4,\Gamma_5,\Gamma_6$, these may be
obtained from $\Gamma_1,\Gamma_2,\Gamma_3$, respectively, by means of
the reflection symmetry, consisting in implementing simultaneously the
two changes of variables, $x\to \mathcal{L}-x$ and
$y\to\mathcal{M}-y$, see Sec.~\ref{sec:symmetries}.

\begin{remark}
It is easily seen that the arctic curve, although continuous
everywhere, together with its first derivative, does not correspond to
a single algebraic curve. Indeed it is only piecewise analytic, with
discontinuities in the second derivative at the six contact points.
The non-analiticity of the arctic curve reflects the fact that the
four-vertex model is not a free-fermion model.
\end{remark}

Concerning our main result, a few comments are in order.  First, we
note that there is a natural bijection between the four-vertex model
with scalar-product boundary conditions and non-intersecting lattice
paths, or lozenge tilings of a hexagon, see Sect.~\ref{sec:bijection}
for details.  Despite this bijection, the four-vertex model is
genuinely non free-fermionic, arising as a particular limit of the
five-vertex model \cite{B-08b,BuP-21}. Indeed the bijection is
somewhat `non-local', in the sense that different portions of the
lattice are deformed in different ways when going back and forth
between the four-vertex model and the non-intersecting lattice
paths. Also, investigation of the limit shape of the four-vertex model
appears beyond current capabilities of the variational approach
(however, see \cite{KP-22} for recent developments).

Second, in view of the above mentioned bijection, it is clear that one
could have deduced the expressions of the six portions of the arctic
curve from suitable transformations of those worked out in
\cite{CLP-99} for the lozenge tilings of a hexagon. It is
nevertheless worth providing an independent derivation, in particular
because the modified version of the Tangent Method proposed here may
be extended to other more sophisticated models, such as, for instance,
the five-vertex model, exhibiting phase separation with regions of
anti-ferroelectric order \cite{dGKW-21}.

Third, our results provide one additional example of the effectiveness
of the Tangent Method for evaluating the arctic curve of various
models of statistical mechanics that have a description in terms of
paths. 

Fourth, the alternative derivation of our main theorem, based on the
evaluation of the asymptotic behaviour of some `gap probability' of
the model, immediately implies, as a by-product, that the fluctuations
of the arctic curve, away from the contact points, are governed by the
Tracy-Widom distribution \cite{TW-94a,TW-94b}.  This follows directly
from the Fredholm determinant representations provided for the gap
probabilities, see \eqref{eq:FD1} and \eqref{eq:FD2} below, and
constitutes one more example in support of the universality of
Tracy-Widom distribution \cite{D-06}.

Finally, one could wonder about the fluctuations of configurations in
the vicinity of the contact points. It would be interesting to verify
whether these are actually described by the Gaussian Unitary Ensemble
corner process, in analogy with lozenge tilings and the ice-model
\cite{OR-06,GP-15,G-14}.

The paper is organized as follows. In Sec.~\ref{sec:properties} we
discuss some useful properties of the model. In
Sec.~\ref{sec:TMmethod} we derive the expression for the arctic curve,
Thm.~\ref{thm:result}, using the Tangent Method. In
Sec.~\ref{sec:EFPmethod} we provide an alternative derivation of the
theorem, using the EFP Method. Additional technical derivations are
given in three Appendices.

\section{Some properties of the model}\label{sec:properties}

\subsection{Reflection and particle-hole symmetries}\label{sec:symmetries}

Let us consider some symmetries of the set of configurations of the
four-vertex model with scalar-product boundary conditions. A first,
elementary symmetry of the model, is that under simultaneous
reflections with respect to the horizontal and vertical axes. In terms
of the lattice coordinates, this transformation reads
\begin{equation}
(n,m)\mapsto(L-n,M-m),
\end{equation}
see Fig.~\ref{fig:reflection}.  The symmetry under this transformation
implies relations for observables of the model and is clearly
preserved in the scaling limit. In particular, concerning the arctic
curve, this makes it possible to determine the portions $\Gamma_4$,
$\Gamma_5$, and $\Gamma_6$ of the arctic curve from the portions
$\Gamma_1$, $\Gamma_2$, and $\Gamma_3$, respectively.

\begin{figure}[t]
\includegraphics[width=.6\linewidth]{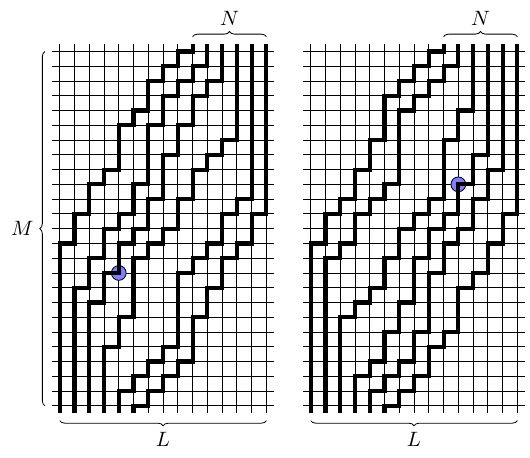}
\caption{A configuration of the four-vertex model (left); the
  resulting configuration after simultaneous reflection with respect
  to both horizontal and vertical axes (right). The blue dot denotes
  the vertex of coordinates $(n,m)$ and its image under the
  transformation. Here $L=15$, $M=25$, $N=6$, $n=5$ and
  $m=10$.}\label{fig:reflection}
\end{figure}

Another symmetry is the following.  If we swap the state (thick/thin)
of each vertical edge and perform a reflection of the model with
respect to an horizontal (or vertical) axis, we end up with a
four-vertex model with the same type of boundary conditions, but with
$L-N$ lines instead of $N$, see Fig.~\ref{fig:particle_hole}. This may
be recognized as the usual particle-hole duality.  The considered
transformation corresponds to the following mapping of the lattice
parameters and coordinates:
\begin{equation}\label{eq:phsym}
L,M,N\mapsto L,M,L-N,\qquad (n,m)\mapsto (L-n+1,m).
\end{equation}
In particular, it follows that the northeast portion, $\Gamma_3$, of
the arctic curve for the model with $N$ lines is readily obtained from
the reflection of the northwest portion, $\Gamma_1$, of the arctic
curve for the model with $L-N$ lines, see \eqref{eq:gamma3}.

\begin{figure}[t]
\includegraphics[width=.85\linewidth]{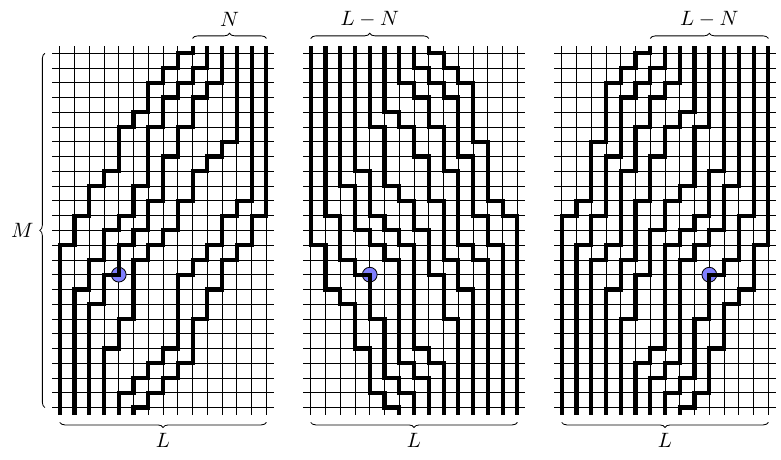}
\caption{A configuration of the four-vertex model (left);
the resulting configurations after swapping the state of all vertical
edges (center), and next performing a reflection (right). The blue dot
denotes the vertex $(n,m)$ and its image under the two
transformations. Here $L=15$, $M=25$, $N=6$, $n=5$ and
$m=10$.}\label{fig:particle_hole}
\end{figure}

\subsection{Equivalent boundary conditions}
It is easily seen that, due to the combinatorial constraints implied
by the vertex rules (absence of the `all thick edges' and of the `two
horizontal thick edges' vertex configurations), four triangular
regions arise in the four corners, where all vertices are frozen.  For
instance, the vertices in the southwest triangular corner, with
lattice coordinates $m+n\leq N$, are all of type $b$; similarly, the
vertices in the northeast triangular corner, with lattice coordinates
$m+n\geq M+L-N+2$, are again all of type $b$.  In the same way, all
vertices in the northwest triangular corner, $m-n\geq M-L+N+1$, and in
the southeast one, $m-n\leq -N-1$, are of type $a$.  Recalling that
all Boltzmann weights are set to $1$, we may cut away these four
triangular corners with no effect on the Boltzmann weights of the
configurations, nor on the partition function. After that we are left
with the model on a hexagonal domain, see Fig.~\ref{fig:hex-bc}.

\begin{figure}[t]
\includegraphics{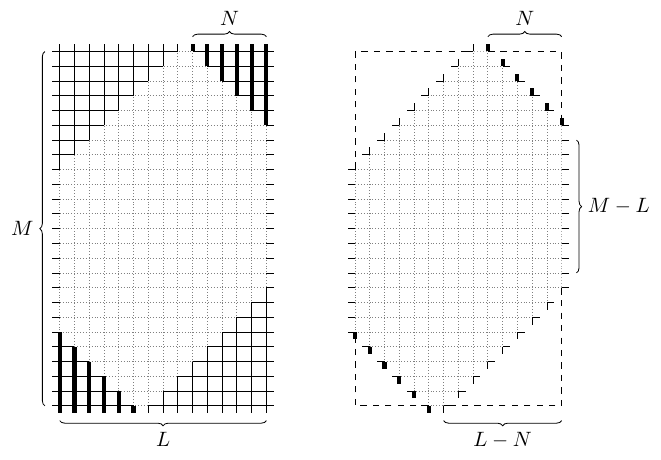}
\caption{The four-vertex model on the rectangular domain, with the
  four frozen triangular regions (left), and on the hexagonal
  domain (right).}\label{fig:hex-bc}
\end{figure}

\subsection{Four-vertex model and non-intersecting lattice
  paths}\label{sec:bijection} There exists a natural bijection between
the configurations of the four-vertex model and non-intersecting
lattice paths, corresponding to that between strictly decreasing and
ordinary boxed plane partitions.  In terms of paths, the bijection is
based on two facts: first, each configuration is uniquely determined
by the positions of the thick horizontal edges, and, second, each path
has exactly $L-N$ such edges. If we shift the $i$-th horizontal thick
edge ($i=1,\dots, L-N$, counting from the southwest) of each path by
$i-1$ lattice spacings southward, then we get a configuration of $N$
non-intersecting lattice paths on the $L\times K$ lattice,
$K=M-L+N+1$, with no further constraint, see Fig.~\ref{fig:NILP}.  In
this new setting, the $j$th path, $j=1,\dots,N$ connects the vertices
of coordinates $(j,1)$ and $(L-N+j,K)$.

\begin{figure}[t]
\includegraphics{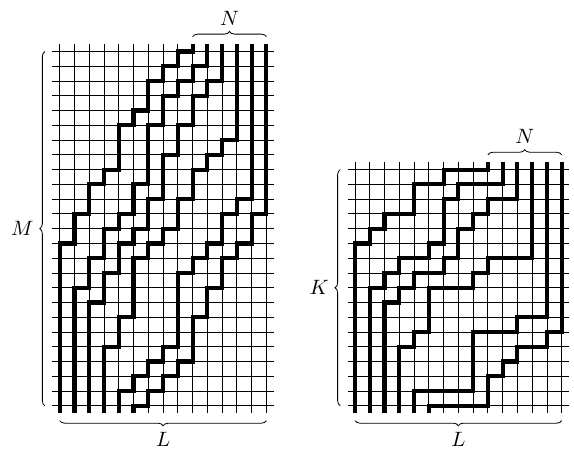}
\caption{A configuration of the four-vertex model (left), and the
  corresponding configuration of non-intersecting lattice paths
  (right). Here $L=15$, $M=25$, $N=6$, and thus $K=17$.}
\label{fig:NILP}
\end{figure}

Note that, in the non-intersecting lattice path version of the model,
due to combinatorial constraints, we may replace the vertical paths in
the southwest and northeast corners of the lattice with horizontal
ones, or may even simply remove these corners alltogether, obtaining
the same result.  In other words, we could have equivalently
considered, as southwest endpoints of the $j$th path, $j=1,\dots,N$,
the vertices $(1,N-j+1)$, or the vertices $(j,N-j+1)$. Similarly, we
could have considered as northeast endpoints the vertices $(L,K-j+1)$,
or the vertices $(L-N+j,K-j+1)$, again obtaining the same result. This
property will be sometimes used for simplicity, in the calculation of
various correlation functions of the model, in the Appendices.

With the help of this bijection, by applying the
Linstr\"om-Gessel-Viennot lemma \cite{L-73,GV-85}, see also
\cite{S-90}, one can find explicit expressions for various quantities
of interest of the model, the simplest being of course the partition
function.

\begin{proposition}\label{prop:Z}
The partition function is given by
\begin{equation}\label{eq:Z=B}
  Z_{L,M,N} = 
        \prod_{j=1}^{N}\frac{(j-1)!(M-N+j)!}{(L-N+j-1)!(M-L+j)!}.
\end{equation}
\end{proposition}
The proof is very standard, see App.~\ref{app:boundary} for details.

\begin{remark} The  above product 
  may be equivalently rewritten as $\mathrm{PL}(N,L-N,M-L+1)$, where
\begin{equation}\label{eq:BPP}
	\mathrm{PL}(r,s,t) = \prod_{i=1}^{r}\prod_{j=1}^{s}\prod_{k=1}^{t}
		\frac{i+j+k-1}{i+j+k-2}
\end{equation}
is the famous Mac-Mahon formula for the number of plane partitions
that fit in a box of size $r \times s \times t$.
\end{remark}

The relation between boxed plane partitions and non-intersecting
lattice paths is well known, see, e.g., \cite{Br-99} and reference
therein.  

\begin{remark} It is clear from the last remark that, in order to have
 more that just one admissible configuration in the model, we must
 have $ M\geq L>N$. 
\end{remark}

\subsection{The boundary-refined partition function}\label{sec:refZ}
A crucial role in the Tangent Method is played by the boundary-refined
partition function, defined as follows. Let us consider a generic
configuration of the system on the hexagonal domain and focus on the
trajectory of the leftmost path. Starting from the southmost vertex
$(1,1)$ such a path will consist of a sequence of $b$- and
$c$-vertices.  There is clearly a last vertex of type $b$. Denoting
its horizontal coordinate by $n$, $n\in[1,L-N+1]$, its vertical
coordinate is $m=n+M-L+N-1$, see Fig.~\ref{fig:boundary-4vm-NILP}.

We may thus define the boundary-refined partition function
$Z_{L,M,N}^{(1)}(n)$, that enumerates the configurations of the model
according to the horizontal coordinate of the last $b$ vertex of the
leftmost path.

\begin{figure}[t]
\includegraphics{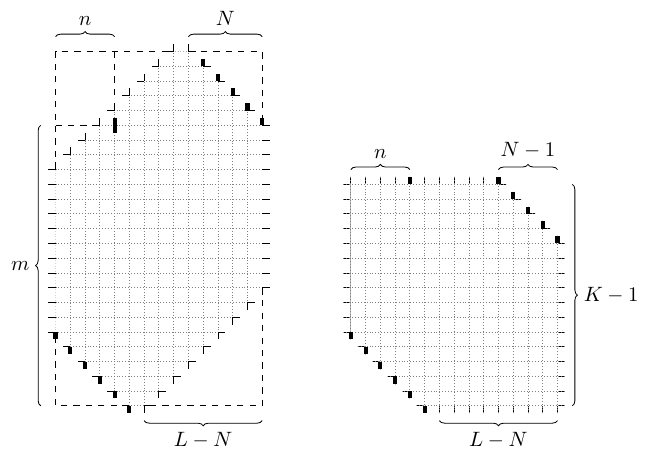}
\caption{The boundary conditions defining the boundary-refined
  partition function $Z_{L,M,N}^{(1)}(n)$ of the four-vertex model
  (left), and the corresponding picture in terms of non-intersecting
  lattice paths (right). Here $L=15$, $M=25$, $N=6$, and $n=5$,
  implying $K=17$ and $m=20$.}
\label{fig:boundary-4vm-NILP}
\end{figure}

\begin{proposition}\label{prop:H}
The number of configurations conditioned to have a vertex of type $b$
at position $(n,n+M-L+N-1)$ is given by
\begin{equation}\label{eq:Zrefined}
  Z_{L,M,N}^{(1)}(n) = \frac{(M-L-1+n)!
    (L-n)!}{(n-1)!(L-n-N+1)!(M-L+N-1)!}Z_{L,M-1,N-1}.
\end{equation}
\end{proposition}
The proof, based on the above mentioned bijection with
non-intersecting lattice paths, is given in
App.~\ref{app:boundary}.

\begin{remark}
It follows from the definition of the partition function and its
boundary-refined counterpart that
\begin{equation}
  \sum_{n=1}^{L-N+1} Z_{L,M,N}^{(1)}(n) =Z_{L,M,N},
\end{equation}
as it may be verified by direct calculation.
\end{remark}

Resorting to the bijection between the four-vertex model and
non-intersecting lattice paths, one may also evaluate more
sophisticated correlation functions. In this respect, a fundamental
building block is what we name `column-refined partition function',
that is the partition function of the model, when conditioned to have
all its thick edges in a given column at assigned positions, see
App.~\ref{app:column-Z} for a precise definition. Partial summations
over the positions of these thick edges allows for the evaluation of
various interesting correlation functions. We provide a derivation of
the column-refined partition function for non-intersecting lattice
paths in App.~\ref{app:column-Z}.  The obtained expression, evidenty
related to discrete log-gases, is used in App.~\ref{app:EFP_AFP} to
evaluate EFP and AFP in the four-vertex model.

\section{Arctic curve via the Tangent Method}\label{sec:TMmethod}

In this section we apply the Tangent Method \cite{CS-16} to obtain the
arctic curve of the model. The method is based on considering a slight
modification of the boundary conditions, and hence of the
configurations of the model. Since the required modification is
different for various portions of the curve, we treat them separately.

\subsection{Determination of the ferroelectric/disorder interface}
To start with, we focus on the northwest portion $\Gamma_1$ of the
arctic curve. To obtain its explicit expression, we slightly modify
the boundary condition by moving the north extremal edge of the
leftmost path from the $(L-N+1)$th vertical edge of the north boundary
to the $r$th one, where $r\leq L-N$, see Fig.~\ref{fig:TM-west}.

\begin{figure}[t]
\includegraphics{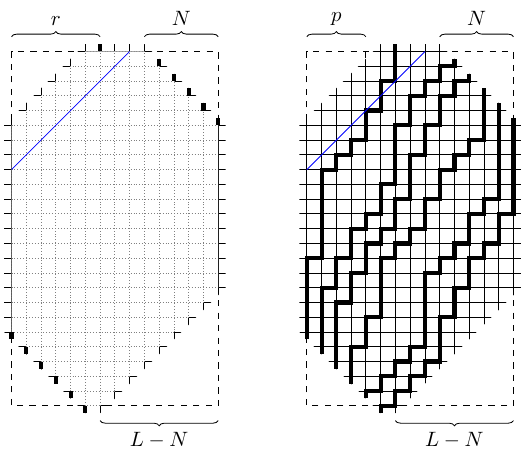}
\caption{Modified boundary conditions, with  $r=7$ (left), and a possible
  configuration, with $p=5$ (right). The blue line indicates
  the northwest boundary of the hexagonal domain.}
\label{fig:TM-west}
\end{figure}

The main assumption of the Tangent Method is that in the scaling limit
such a path first follows the arctic curve, and then departs
\emph{tangentially} from it and becomes a straight line that
intersects the north boundary of the domain at the point corresponding
to (the rescaled value of) $r$.  Note that despite its intuitiveness,
such a behaviour is still conjectural in general. Rigorous results
have been derived in some particular cases \cite{A-20,DGR-19}.  Later
on we shall show that the conjectural arctic curve obtained through
the Tangent Method, besides being well supported by numerical
simulations, may also be derived rigorously by performing the
asymptotic analysis of some suitable correlation functions.

Let us denote the partition function of the model with the modified
boundary conditions by $\tilde{Z}_{L,M,N}^{(1)}(r)$. We can split the
corresponding modified domain into two parts: the hexagonal domain and
the northwest triangular region.  For any given configuration on the
modified domain, let us denote by $p$ the horizontal coordinate of the
vertex where the leftmost path, arriving from southwest, first touches
the border $m=n+M-L+N$. Note that this implies the presence of a
$b$-vertex at position $(p,p+M-L+N-1)$, see Fig.~\ref{fig:TM-west}.
With this setting, the partition function may be expressed as
\begin{equation}\label{eq:TMwZ=sum}
	\tilde{Z}_{L,M,N}^{(1)}(r) = \sum_{p=1}^{r} Z_{L,M,N}^{(1)}(p)
        E_{L,M,N}^{(1)}(p,r),
\end{equation} 
where $Z_{L,M,N}^{(1)}(p)$ is the boundary-refined partition function
defined in Sect.~\ref{sec:refZ}, and $E_{L,M,N}^{(1)}(p,r)$ is the partition
function of the triangular extension, with a path connecting the
vertices of coordinates $(p,p+M-L+N)$ and $(r,M)$, see
Fig.~\ref{fig:TM-wc-split}.

\begin{figure}[t]
\includegraphics{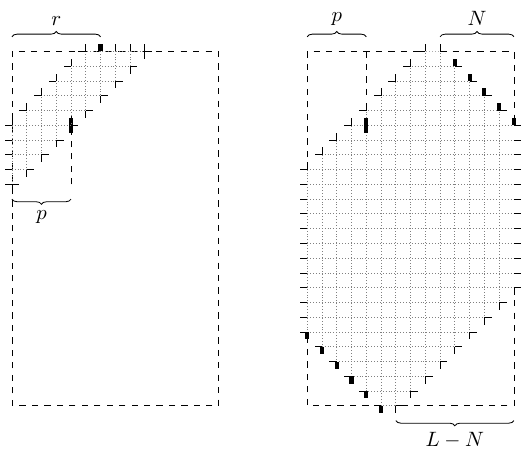}
\caption{The boundary conditions on the triangular extension,
  corresponding to $E_{L,M,N}^{(1)}(p,r)$ (left) and on the hexagonal
  domain, corresponding to the boundary-refined partition functions
  $Z_{L,M,N}^{(1)}(p)$ (right).}\label{fig:TM-wc-split}
\end{figure}

The quantity $E_{L,M,N}^{(1)}(p,r)$ is easily evaluated, being simply
the number of configurations for a single path from $(p,p+M-L+N)$ to
$(r,M)$, satisfying the four-vertex rules, or equivalently, the number
of configurations for a directed lattice path from $(p,M-L+N+1)$ to
$(r,M-r+1)$. One has
\begin{equation}\label{eq:E_LMN_1}
	E_{L,M,N}^{(1)}(p,r) = \binom{L-N-p}{r-p}.
\end{equation}
As for $Z_{L,M,N}^{(1)}(p)$, we resort to Prop.~\ref{prop:H}.

We are interested in studying $\tilde{Z}_{L,M,N}^{(1)}(r)$ in the
scaling limit. To this aim, we rescale the parameters as in
\eqref{eq:scaling}, and also
\begin{equation}\label{eq:scaling2}
  p=\lceil t\ell\rceil, \qquad  r=\lceil  u\ell \rceil ,
\end{equation}
with $t\in(0,u]$ and $u\in(0,\mathcal{L}-\mathcal{N}]$, and send
    $\ell\to\infty$. Expression \eqref{eq:TMwZ=sum} may be interpreted
    as a Riemann sum.  In the large $\ell$ limit, this becomes an
    integral, and can be evaluated by the saddle-point method. We are
    thus led to define the `action'
\begin{align}\nonumber
  S^{(1)}(t,u):=&\lim_{\ell\to\infty}\frac{1}{\ell}\log\left[
    \frac{Z_{L,M,N}^{(1)}(p)E_{L,M,N}^{(1)}(p,r)}{Z_{L,M,N}}\right] \\ =&
  (\mathcal{L}-t) \log( \mathcal{L} - t)+ (\mathcal{M}-\mathcal{L}+t)
  \log (\mathcal{M} - \mathcal{L} + t) \nonumber \\&\qquad - (u-t)\log
  ( u-t ) - t \log t +\dots ,
\end{align}
where, in the last line, we have dropped terms that do not depend on
$t$. The saddle-point equation reads
\begin{equation}
  0=\frac{\dd\ }{\dd t}S^{(1)}(t,u)=
  \log\frac{(\mathcal{M}-\mathcal{L}+t)(u-t)}{t(\mathcal{L}-t)}.
  \end{equation}
It is easily verified that its unique solution,
\begin{equation}\label{eq:SPsol}
  t_0=\frac{\mathcal{M}-\mathcal{L}}{\mathcal{M}-u}u,
\end{equation}
corresponds to a maximum of $S^{(1)}(t,u)$. We have thus shown that,
for any given value of $u\in(0,\mathcal{L}-\mathcal{N}]$, the leftmost
  path exits the hexagonal domain at some value $t$ that concentrates
  around $t_0$.  We conclude that in the scaling limit the upper
  portion of such path follows the straight line through the points
  $(t_0, \mathcal{M-L+N}+t_0)$ and $(u,\mathcal{M})$.

Let us now proceed with the derivation of the curve. According to the
Tangent Method the west arc is the envelope of the one-parametric
family of lines passing through the two points specified above, that
is,
\begin{equation}
  (y-\mathcal{M})(t_0-u)=(x-u)(\mathcal{N}-\mathcal{L}+t_0),\qquad
  u\in(0,\mathcal{L-N}],
\end{equation}
or, using \eqref{eq:SPsol} to eliminate $t_0$,
\begin{equation}\label{eq:family1}
  u(u-\mathcal{L})y+[\mathcal{M}(\mathcal{L}-\mathcal{N}-u)+\mathcal{N}u]x
  +\mathcal{N}(\mathcal{M}-u)u=0,
  \qquad u\in(0,\mathcal{L-N}].
\end{equation}
The envelope is easily evaluated by differentiating with respect to
$u$,
\begin{equation}\label{eq:family2}
  (\mathcal{L}-2u)y+(\mathcal{M}-\mathcal{N})x-
  \mathcal{N}(\mathcal{M}-2u)=0,
\end{equation}
and solving the obtained pair of equations for  $x$ and $y$,
with the result
\begin{align}
  x(u)=&\frac{\mathcal{N}(\mathcal{M-L})u^2}
  {(\mathcal{M-N})u^2-2\mathcal{M}(\mathcal{L-N})u +\mathcal{LM(L-N)}},
  \\
  y(u)=&\frac{\mathcal{N(M-N)}u^2-2\mathcal{MN(L-N)}u+\mathcal{M}^2
    \mathcal{N(L-N)}}{(\mathcal{M-N})u^2-2\mathcal{M}(\mathcal{L-N})u
    +\mathcal{LM(L-N)}}.
\end{align}
The arc of the curve described by this last expression, as $u$ varies over
the interval $(0,\mathcal{L-N})$, is the northwest portion of the
arctic curve, $\Gamma_1$. Indeed it is easily verified that
$(x(0),y(0))$ lies on the line $x=0$, that is the west boundary of the
hexagonal domain, and that $(x(\mathcal{L-N}),y(\mathcal{L-N}))$ lies
on the line $y=x+\mathcal{M-L+N}$, that is the northwest boundary of
the hexagonal domain. Finally, note that $x(u)\vert_{u=\mathcal{L-N}}=\xc$,
with $\xc$ given by \eqref{eq:xc}.

The portion $\Gamma_1$ of the arctic curve may also be written in
implicit form, by using \eqref{eq:family2} to eliminate $u$ in
\eqref{eq:family1}. It is then given by the portion of the algebraic
curve (actually, an ellipse)
\begin{equation}\label{eq:ellipse}
	\left(\mathcal{L} y - \mathcal{NM} + (\mathcal{M-N})x \right)^2 + 
	4 \mathcal{M}(\mathcal{L}-\mathcal{N})x(\mathcal{N}-y) = 0
\end{equation}
lying in the region $x\in[0,\xc]$, $y\in[y(0),y(\xc)]$, with
$y(0)=\mathcal{M}\mathcal{N}/\mathcal{L}$.

Solving the equation \eqref{eq:ellipse} for $y$ we find
that
\begin{equation}
	y =
	\frac{\mathcal{MN}(\mathcal{L}-2x)+
          (\mathcal{M+N})\mathcal{L}x}{\mathcal{L}^2}
	\pm 
	\frac{ 2\sqrt{  \mathcal{MN(M-L)(L-N)} (\mathcal{L}-x) x  }  } 
			 {\mathcal{L}^2}.
\end{equation}
Since we are interested in the northwest portion of the arctic curve,
we must choose the plus sign. We have thus obtained the expression
\eqref{eq:ac-answer1} stated in Thm.~\ref{thm:result}. Note that our
calculation here allows to determine only that portion of the arctic
curve with $x\in (0,\xc]$, corresponding to values of the slope of the
  tangent path decreasing from $+\infty$ to $1$.

\subsection{Determination of the anti-ferroelectric/disorder interface}
We now turn to the determination of the north portion, $\Gamma_2$ of
the arctic curve. Note, however, that the frozen region adjacent to
$\Gamma_2$ is fully packed of $c$ vertices, that is of zigzag paths,
and the Tangent Method may not be applied directly, but requires a
slightly involved preparation.

First we extend the hexagonal domain by adding a triangular region to
its northeast side, consisting in the vertices of coordinates
$\{(n,m)$: $L-N<n\leq L$, $\abs{m-M-1}< n-L+N\}$, and impose all
vertices on its east boundary to be of type $c$ (starting with a
vertex of the last type listed in Fig.~\ref{fig:weights}, and
alternating between the two types of $c$-vertices, see
Fig.~\ref{fig:TM-north-1}).  Such boundary condition constrains all
vertices in the triangular extension to be of type $c$, with the $N$
paths from the hexagonal domain being continued through the triangular
extension with a zig-zag pattern, till the east boundary of the
extended domain. The extension being frozen by construction, the
configurations in the hexagonal domain and the partition function are
left unaltered.

\begin{figure}[t]
\includegraphics{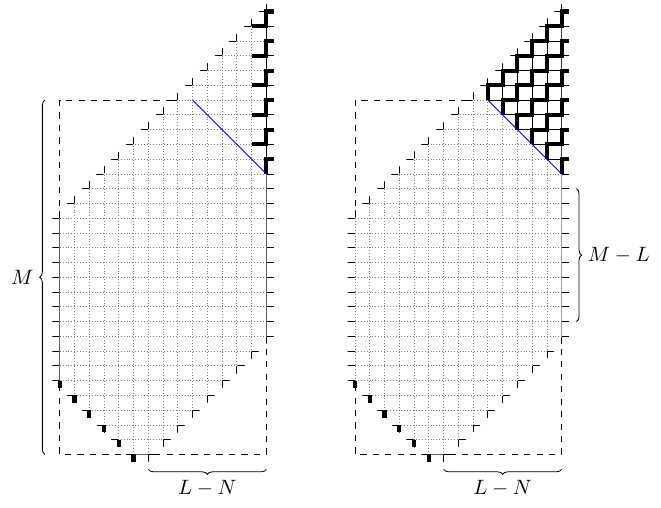}
\caption{The hexagonal domain with the triangular extension, and the
  new boundary conditions (left); these induce freezing of the whole
  triangular region, while preserving the boundary conditions of the
  original hexagonal domain (right). The blue line indicates the
  nortwest boundary of the hexagonal domain.}\label{fig:TM-north-1}
\end{figure}

Next we modify the boundary conditions on the east side by adding one
vertex of type $a$ at position $(L,M+N-2r)$, $r=0,\dots,N$, and
shifting the $2N-2r-1$ remaining vertices one position to the south,
see Fig.~\ref{fig:TM-north-2a}. The meaning of this modification
becomes clear when one focuses on the vertices of type $a$ and
interprets them as particles. Then the configurations of the model are
essentially configurations of two families of non-intersecting lattice
paths, see Fig.~\ref{fig:TM-north-2b}. The topmost thin (blue) path is
the relevant one for application of the Tangent Method in this
extended context, see Fig.~\ref{fig:TM-north-2b}, right.  For the sake
of clarity, we have also extended the domain with two triangular
regions of $a$ vertices.

\begin{figure}[t]
\includegraphics{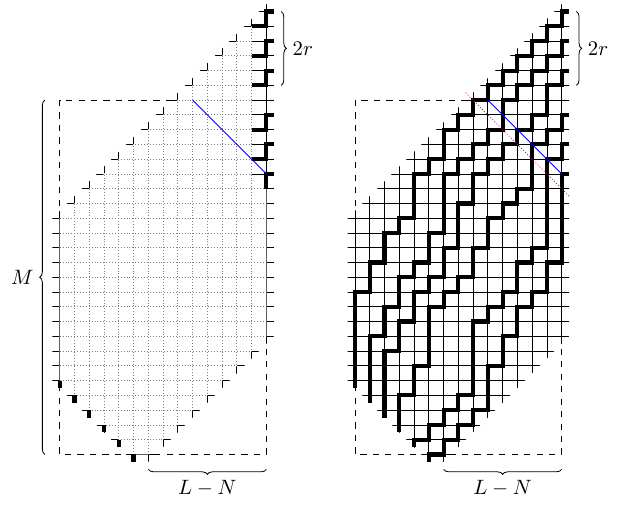}
\caption{Modified boundary conditions, with $r=3$ (left), and a
  possible configuration (right).  The blue line indicates the
  northwest boundary of the hexagonal domain. There is one and only
  one vertex of type $a$ lying on the dotted red line.}
  \label{fig:TM-north-2a}
\end{figure}

\begin{figure}[t]
\includegraphics{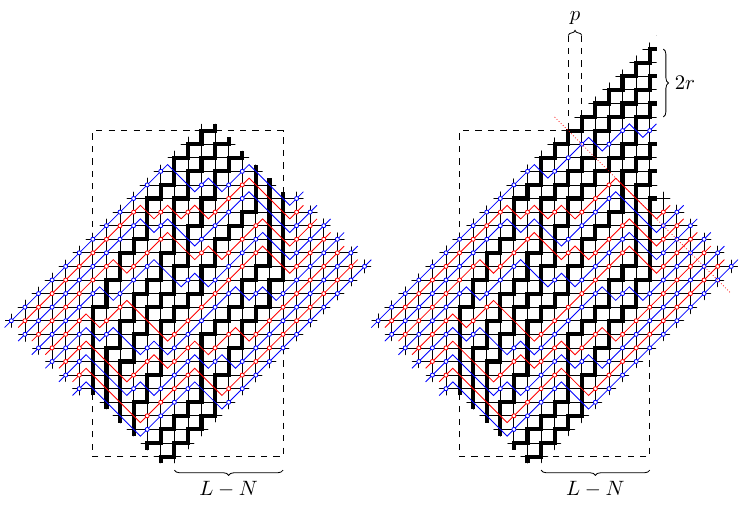}
\caption[]{Configuration of the four-vertex model with original (left)
  and modified (right) boundary conditions in terms of two species of
  non-intersecting lattice paths. The topmost thin (blue) path in the
  right picture is the relevant one for the Tangent Method. Here $p=1$
  and $r=3$. The dotted red line is as in
  Fig.~12.}\label{fig:TM-north-2b}
\end{figure}

Let us evaluate the quantity $\tilde{Z}_{L,M,N}^{(2)}(r)$, denoting
the partition function of the model with the modified boundary
conditions. Again, we may split the domain into two parts: the
original hexagonal domain, and its triangular extension.
For any given configuration on the extended domain with modified
boundary conditions, there is one and only one vertex of type $a$
lying on the line $m+n=L-N+M$ (adjacent to the northeast boundary of
the hexagonal domain, see Fig.~\ref{fig:TM-north-2a}, right). We may
write the coordinates of such an $a$-vertex as $(L-N+p,M-p)$, and thus
use $p=0,\dots,N$ to label its position.  With this setting, the
partition function of the model with modified boundary conditions may
be expressed as
\begin{equation}\label{eq:Z2}
	\tilde{Z}^{(2)}_{L,M,N}(r) = \sum_{p=0}^{r} Z^{(2)}_{L,M,N}(p)
        E^{(2)}_{L,M,N}(p,r),
\end{equation}
where $Z^{(2)}_{L,M,N}(p)$ is the partition function on the hexagonal
domain, conditioned to have a vertex of type $a$ at position
$(L-N+p,M-p)$, and $E^{(2)}_{L,M,N}(p,r)$ is the partition function of
the triangular extension.  Referring to Fig.~\ref{fig:TM-north-2b},
right, focussing on the top blue path therein, and proceeding as for
$E_{L,M,N}^{(1)}(p,r)$, see \eqref{eq:E_LMN_1}, it is easy to find
that
\begin{equation}
	E^{(2)}_{L,M,N}(p,r) = \binom{ N-p}{r-p}.
\end{equation}
As for the partition function $Z^{(2)}_{L,M,N}(p)$, using the
particle-hole symmetry \eqref{eq:phsym}, we find that
\begin{equation}
	Z^{(2)}_{L,M,N}(p) = Z_{L,M,L-N}^{(1)}(N-p+1),
\end{equation}
where $Z_{L,M,N}^{(1)}(n)$ is the boundary-refined partition function,
evaluated in \eqref{eq:Zrefined}.

We may now proceed to evaluate \eqref{eq:Z2} in the scaling limit.
Let us rescale the parameters and coordinates of the lattice variables
as in \eqref{eq:scaling}. Let also $p=\lceil t\ell\rceil$ and
$r=\lceil u\ell\rceil$, with $t\in[0,u]$ and $u\in[0,\mathcal{N}]$.
Again, in the scaling limit, the sum in \eqref{eq:Z2} turns into an
integral whose values is dominated by the saddle-point of the `action'
\begin{align}
  S^{(2)}(t,u):=&\lim_{\ell\to\infty}\frac{1}{\ell}
  \log\left[
    \frac{(M-L+N-p)!(L-N+p)!}{p!(N-p)!}
    \binom{N-p}{r-p}
    \right]
  \\
  =&
  (\mathcal{M-L+N}-t)\log(\mathcal{M-L+N}-t)+
  (\mathcal{L}-\mathcal{N}+t)\log (\mathcal{L}-\mathcal{N}+t)
  \\
  &\qquad
  - t\log t - (u-t)\log(u-t)-  (\mathcal{N} - u) \log(\mathcal{N} - u),
\end{align}
namely
\begin{equation}\label{eq:sps2}
	t_{0}= 
		\frac{ u ( \mathcal{L} - \mathcal{N} ) }
				 { \mathcal{M} -u}.
\end{equation}
It is straigthforward to check that $t_0$ is the only stationary point
of $S^{(2)}(t,u)$, that it actually corresponds to a maximum, and that
it lies inside the interval $(0,u)$ for all values of $u\in(0,\mathcal{N})$.

Therefore in the scaling limit the diverted path becomes a straight
line through $(\mathcal{L-N}+t_0,\mathcal{M}-t_0)$ and
$(\mathcal{L},\mathcal{M+N}-2u)$.  The arctic curve is the envelope of
the family of such lines as $u\in[0,\mathcal{N}]$,
\begin{equation}
	y =
		\frac{ x - \mathcal{L} }
				 { t_0 - \mathcal{N} }
		(\mathcal{M} - t_0)
		+
		\frac{x - (\mathcal{L}-\mathcal{N}+t_0)}
				 {\mathcal{N} - t_0}
		(\mathcal{M}+\mathcal{N}-2u),
\end{equation}
or, substituting \eqref{eq:sps2} for $t_0$, 
\begin{multline}\label{eq:y=y(x)TM-north}
	\mathcal{MN} (\mathcal{M-L+N}-y+x-2u)
	+ \mathcal{L} u (\mathcal{M-L+N} + x + y)
	\\
	+ 2u x (u  - \mathcal{M-N}) = 0.
\end{multline}
Along the line of previous Section, differentiating the above
expression with respect to $u$, and solving for $u$, we get
\begin{equation}
	u =\frac{ \mathcal{L}^2 + 2 \mathcal{N}x +
          2\mathcal{M}(\mathcal{N} + x) -
          \mathcal{L}(\mathcal{M}+\mathcal{N}+x+y)} { 4x }.
\end{equation}
Substituting this in  \eqref{eq:y=y(x)TM-north} and solving for $y$, we
get
\begin{multline}
	y = \mathcal{L-M-N} - x + 2\,\frac{\mathcal{NM} (\mathcal{L}-2x)+
          (\mathcal{M+N})\mathcal{L}x}{\mathcal{L}^2}
        \\
        \pm 4 \,\frac{ \sqrt{ \mathcal{MN(M-L)(L-N)} (\mathcal{L}-x) x } } {
          \mathcal{L}^2 }.
\end{multline}
This curve is again an ellipse, although a different one, with respect
to \eqref{eq:ellipse}. Since we are interested in the north portion of
the curve, we must choose the plus sign, thus obtaining the expression
\eqref{eq:ac-answer2} stated in Thm.~\ref{thm:result}.

\section{Arctic curve via the EFP Method}\label{sec:EFPmethod}

In this section we introduce two observables which we call `Emptiness
Formation Probability' (EFP) and `Anti-ferroelectric phase Formation
Probability' (AFP). We express these functions in terms of gap
probabilities for discrete log-gases associated to the measure of Hahn
polynomials. Then, following \cite{J-02}, and using the results of
\cite{BKMM-07} for the asymptotic behaviour of these quantities in the
scaling limit, we derive an expression for the arctic curve, and
determine its fluctuations.

\subsection{EFP and AFP}
Let us consider, in the rectangular domain, the topmost $q$
consecutive vertices in the $p$th vertical line, $p\leq L-N$, that is
the vertices of coordinates $(p,M-j+1)$, $j=1,\dots,q$.  Let $\Theta$
be the subset of such vertices contained in the hexagonal domain, that
is the vertices of coordinates $(p,M-q+j)$, $j=1,\dots, \tilde{q}$,
with $\tilde{q}:=p+q-L+N$.  Let $\Omega$ be the set of all
configurations of the model, and $\Omega^{\mathrm{I}}_{p,q}$,
$\Omega^{\mathrm{II}}_{p,q}$ the sets of configurations in which all
vertices in $\Theta$ are of type $a$, or $c$, respectively. We may now
introduce two observables, the EFP and the AFP, denoted as
$F_{L,M,N}(p,q)$ and $G_{L,M,N}(p,q)$, respectively, and defined as
\begin{equation}
  F_{L,M,N}(p,q) :=
  \frac{ \abs{\Omega^{\mathrm{I}}_{p,q}} }{\abs{\Omega}},\qquad
  G_{L,M,N}(p,q) :=
  \frac{ \abs{\Omega^{\mathrm{II}}_{p,q}} }{\abs{\Omega}},
\end{equation}
where $\abs{\Omega}$ stands for the cardinality of $\Omega$.

Note, that, in view of the chosen boundary conditions, the request of
having all vertices in $\Theta$ of type $a$ (for EFP) or $c$ (for AFP)
implies the existence of a subregion in the hexagonal lattice where
all vertices are of type $a$, or $c$, respectively, see
Fig.~\ref{fig:EFPandAFP}.

In what follows we will always assume that
\begin{equation}\label{eq:EFP-cond1}
  p + q > L - N, 
\end{equation}
which guarantees that $\Theta$ is non empty, and also
\begin{equation}\label{eq:EFP-cond2}
  p \leq  L - N, \qquad p + q <M - N,
\end{equation}
for EFP, and  
\begin{equation}\label{eq:EFP-cond3}
  p + q > L - N, \qquad  p \leq  L - N,  \qquad p+q\leq M-N
\end{equation}
for AFP, guaranteeing that $\Omega^{\mathrm{I}}_{p,q}$ and
$\Omega^{\mathrm{II}}_{p,q}$, respectively, are non empty; if this
were the case, the corresponding formation probabilities would simply
vanish.

\begin{figure}[t]
\includegraphics{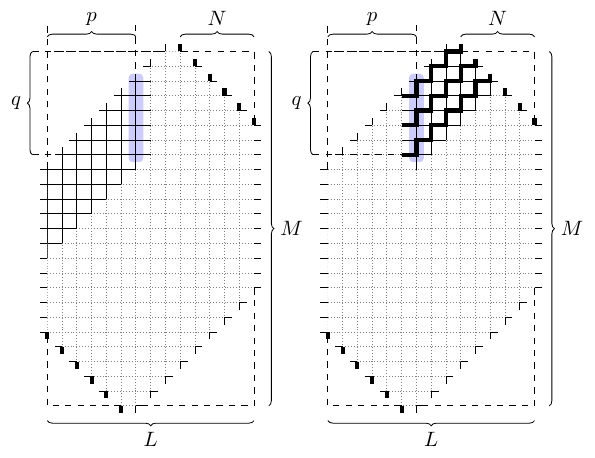}
\caption{Subregions of the hexagonal domain that are frozen, due to
  the condition imposed on the vertices in $\Theta$ (shaded), in the
  case of EFP (left) and AFP (right). Here $L=15$, $M=25$, $N=6$,
  $p=7$, $q=8$, and $\tilde{q}=6$.}
\label{fig:EFPandAFP}
\end{figure}

\subsection{Log-gas representations}
We now turn to the evaluation of some representations of EFP and AFP
in terms of discrete log-gases \cite{F-10}.  To start with, let us
introduce the Hahn  measure,
\begin{equation}\label{eq:wHahn=}
	w^{(\alpha,\beta)}_n(x) = \binom{\alpha+x}{x}\binom{\beta+n-x}{n-x},
\end{equation}
defined on the discrete interval $x\in[0,n]$. When the parameters
$\alpha,\beta >-1$ or $\alpha,\beta<-n$, we may define the
corresponding family of orthogonal polynomials $\{
Q^{(\alpha,\beta)}_{k,n}(x)\}_{k=0,\dots,n}$, satisfying the
orthonormality condition
\begin{equation}
	\sum_{x=0}^n w^{(\alpha,\beta)}_n(x)
        Q^{(\alpha,\beta)}_{k,n}(x) Q^{(\alpha,\beta)}_{l,n}(x)=\delta_{k,l}.
\end{equation}
The polynomials
\begin{multline}\label{eq:orthoQ}
	Q^{(\alpha,\beta)}_{k,n}(x) = 
		(-1)^{k}
		\sqrt{\binom{n}{k}
		\frac{n!(2k+\alpha+\beta+1)(\alpha+1)_k(\alpha+\beta+1)_k}
			{(\alpha+\beta+1)_{n+1}(\beta+1)_k(n+\alpha+\beta+2)_k}
	}\\
	\times
		\FPQ{-k ,\, k+\alpha+\beta+1, \, -x }{\alpha+1,\, -n}{1},\qquad
                k=0,\dots,n, 
\end{multline}
are known as (normalized) Hahn polynomials
\cite{KLS-10}.

We are now ready to introduce the Hahn log-gas. Let
$\mathbf{x}:=\{x_1,\dots,x_s\}$, with $0\leq x_1<\dots<x_s\leq n$,
denote the ordered set of positions of $s$ particles on the discrete
interval $[0,n]$. Consider the probability measure on $[0,n]^s$,
defined as
\begin{equation}\label{eq:P[x]=}
	P^{(\alpha,\beta)}_{n,s} [\mathbf{x}] = 
		\frac{1}{Z(\alpha,\beta,s,n)} \prod_{1\le i < j \le s} (x_i-x_j)^2
		\prod_{j=1}^{s} w^{(\alpha,\beta)}_n (x_j)
\end{equation}
where $w^{(\alpha,\beta)}_n (x)$ is the Hahn measure
\eqref{eq:wHahn=}. The normalization constant
\begin{equation}
	Z(\alpha,\beta,s,n)= \sum_{0\leq \mathbf{x}\leq n} \prod_{1\le
          i < j \le s} (x_i-x_j)^2 \prod_{j=1}^{s}
        w_n^{(\alpha,\beta)}(x_j)
\end{equation}
is the partition function of the log-gas.

\begin{remark}
The partition function $ Z(\alpha,\beta,s,n)$ may also be represented
as the determinant of a Hankel matrix,
\begin{equation}\label{eq:Z-Hankel}
  Z(\alpha,\beta,s,n)=\det_{1\leq i,j\leq s}
  \left[\sum_{x=0}^n x^{i+j-2} w_n^{(\alpha,\beta)}(x)\right],
\end{equation}
whose elements are the moments of Hahn measure. In principle, this
could allow for a direct evaluation of the partition functions in
terms of the leading coefficients of orthonormal Hahn polynomials.
\end{remark}

Let us introduce also the function
\begin{equation}\label{eq:H=def}
  H(d,\alpha,\beta,s,n) :=\sum_{0\leq \mathbf{x}\leq d}
  P^{(\alpha,\beta)}_{n,s} [\mathbf{x}],
\end{equation}
with $d\leq n$, which is nothing but the `gap probability', that is,
the probability of having, in the Hahn log-gas of $s$ particles, no
particle with coordinate larger that $d$. Clearly, the gap probability
may be rewritten as a Hankel determinant, just as done above for the
partition function, see \eqref{eq:Z-Hankel}.

We may now state three propositions expressing EFP and AFP in terms of
discrete Hahn log-gases. The proof of these propositions, based on the
bijection from the four-vertex model to non-intersecting lattice
paths, is provided in App.~\ref{app:EFP_AFP}.

\begin{proposition}\label{prop:EFP}
EFP $F_{L,M,N}(p,q)$ is given by
\begin{equation}\label{eq:F(p,q)=}
	F_{L,M,N}(p,q) = H(d,\alpha,\beta,s,n)
\end{equation}
with parameters
\begin{gather}\label{eq:F(p,q)parameters=}
  d = M-N+\min(p,N)-p-q, \qquad \alpha = \abs{N-p}, \qquad \beta =
  L-N-p,\\ s =\min(p,N), \qquad n = M-L+\min(p,N).
\end{gather}
\end{proposition}

\begin{remark}
The conditions $\alpha> -1$ and $\beta> -1$, entering the definition
of the orthogonalizing measure for Hahn polynomials, are always
fulfilled here. This is evident for $\alpha$; as for $\beta$, recall
that when $p>L-N$, $F_{L,M,N}(p,q)=0$ by construction, see
\eqref{eq:EFP-cond2}.
\end{remark}

\begin{proposition}\label{prop:EFP_AFP}
  AFP is related to EFP as follows:
\begin{equation}
  G_{L,M,N}(p,q) =
  F_{M-L+N+1,M,M-L+1}\left(\tilde{r},L-\tilde{p}\right),
\end{equation}
where $\tilde{r}=\big\lfloor\frac{\tilde{q}+1}{2}\big\rfloor$,
$\tilde{p}=p-\mathrm{mod}(\tilde{q}+1,2)$, and $\tilde{q}=p+q-L+N$.
\end{proposition}

In turn, these imply the following explicit representation for AFP.

\begin{proposition}\label{prop:AFP} AFP $G_{L,M,N}(p,q)$ is given by
\begin{equation}\label{eq:G(p,q)=}
	G_{L,M,N}(p,q) = H(d,\alpha,\beta,s,n)
\end{equation}
with parameters
\begin{gather}\label{eq:G(p,q)parameters=}
  d = L-N+\min(\tilde{r}, M-L + 1)- 2 - q +\tilde{r},\\
	\alpha = \abs{M - L - \tilde{r} + 1},\qquad
	\beta = N-\tilde{r},\\
	s =  \min(\tilde{r}, M-L + 1),\qquad
	n = L-N+\min(\tilde{r}, M-L + 1)-1,
\end{gather}
where
$\tilde{r}=\big\lfloor\frac{\tilde{q}+1}{2}\big\rfloor=\big\lfloor\frac{p+q-L+N+1}{2}\big\rfloor$.
\end{proposition}

\begin{remark}
The conditions $\alpha> -1$ and $\beta> -1$, entering the definition
of the orthogonalizing measure for Hahn polynomials, are always
fulfilled here as well. Again, this is evident for $\alpha$; as for
$\beta$, recall that when $\tilde{r}>N$, $G_{L,M,N}(p,q)=0$ by
construction, see \eqref{eq:EFP-cond3}.
\end{remark}

\subsection{Arctic curve}
To find the arctic curve we need to study the asymptotic behaviour of
EFP and AFP, or equivalently of the function $H(d,\alpha,\beta,s,n)$,
when $d,\alpha,\beta,s,n\to\infty$ with fixed ratios. In such a limit,
$H(d,\alpha,\beta,s,n)$ is known to tend to one if the values of the
parameters are such that the vertices in $\Theta$ are all in a frozen
region, and to zero otherwise. The arctic curve is characterized by
separating these two regimes, see, e.g., \cite{J-00,J-02}.

In order to analize the asymptotic behaviour of
$H(d,\alpha,\beta,s,n)$ in the desired scaling limit, one may resort
to various technologies inspired from and related to random matrix
models and the theory of orthogonal polynomials.

To start with, we set
\begin{equation}\label{scaling2}
        d=\lfloor d_0 \ell\rfloor,
    \qquad \alpha=\lfloor\alpha_0 \ell\rfloor,
    \qquad \beta=\lfloor\beta_0 \ell\rfloor,
    \qquad s=\lceil s_0 \ell\rceil,
    \qquad n=\lceil n_0 \ell\rceil,
\end{equation}
with $\alpha_0,\beta_0>0$ and $s_0<d_0<n_0$. Following standard
methods from the theory of random matrix models, we rescale the
coordinates, $x_j=\lfloor\mu_j \ell\rfloor$, $j=1,\dots,s$.  The sums
over $x_j$'s may be interpreted as Riemann sums, and, in the large
$\ell$ limit, replaced by integrals. Correspondingly, we introduce the
density $\rho(\mu)$, such that
\begin{equation}
	\int_0^{n_0} \rho(\mu) \dd \mu = s_0, \qquad 
	0\leq\rho(\mu) \le 1.
\end{equation}
The first equation is the usual normalization condition, while the
second one is an additional constraint following from the discreteness
of the original variables $x_j$ \cite{DK-93}.

In general, the density function $\rho(\mu)$ may be computed either by
directly solving a suitable variational problem \cite{BKMM-07} or by
resorting to an integral representation based on the asymptotics of
the recursion coefficients for the associated orthogonal polynomials
\cite{KA-99}.  When the primary interest lies in determining the
sole support of the density, rather than obtaining its complete
expression, the latter approach proves to be much more efficient, see
e.g. \cite{J-02}.  This is precisely the case when determining the
arctic curve.  Indeed, denoting by $R(\alpha_0,\beta_0,s_0,n_0)$ the
right endpoint of the support of the density, the arctic curve is
determined by the condition
\begin{equation}\label{eq:R=n-d0}
	R(\alpha_0,\beta_0,s_0,n_0) = d_0,
\end{equation}
see, e.g., \cite{J-00}.

Hahn polynomials have been extensively studied in the literature,
therefore to get the expression of the corresponding density function
we resort to \cite{BKMM-07}. Specifically, for the right endpoint of
the support of the density, see Eq.~(276) therein. In our notations,
it reads:
\begin{multline}\label{right-end-point}
  R(\alpha_0,\beta_0,s_0,n_0) =
  \\
  = \left( \frac{
    \sqrt{(s_0+\alpha_0+\beta_0)(s_0+\alpha_0)(n_0-s_0)} +
    \sqrt{(s_0+\alpha_0+\beta_0+n_0)(s_0+\beta_0)s_0}}
  {(2s_0+\alpha_0+\beta_0)}\right)^2.
\end{multline}

Inserting now \eqref{right-end-point} into \eqref{eq:R=n-d0}, and
specializing the parameters $\alpha_0$, $\beta_0$, $s_0$, $n_0$, $d_0$
according to Prop.~\ref{prop:EFP}, we recover expression
\eqref{eq:ac-answer1} for the portion $\Gamma_1$ of the arctic
curve. In a similar way, specializing the parameters $\alpha_0$,
$\beta_0$, $s_0$, $n_0$, $d_0$ according to Prop.~\ref{prop:AFP}, we
recover expression \eqref{eq:ac-answer2} for the portion $\Gamma_2$ of
the arctic curve.

\subsection{Arctic curve fluctuations}
Having derived the arctic curve of the four-vertex model, the natural
question to address next is that of its fluctuations. In view of the
close relation of the model with non-intersecting lattice paths, or
with lozenge tilings of a hexagon, we expect these fluctuations to be
again governed by Tracy--Widom distribution
\cite{TW-94a,TW-94b}. Indeed, having expressed EFP and AFP as gap
probabilities for the discrete log-gas with the Hahn measure, the answer
to this question follows from the results of \cite{BKMM-07}, see also
\cite{J-00,J-02}.

To start with, we recall a very standard fact, namely that the `gap
probability' $H(d,\alpha,\beta,s,n)$ may be rewritten as a Fredholm
determinant:
\begin{equation}\label{eq:FD1}
  H(d,\alpha,\beta,s,n)=\det\left[1-K_{n,s}|_{(d,n]}\right],
\end{equation}
where $K_{n,s}|_{(d,n]}$ is a discrete integral operator acting on $L^2(d,n]$
    with kernel 
\begin{equation}\label{eq:FD2}
  K_{n,s}(x,y)=\sum_{k=0}^sQ_{k,n}^{(\alpha,\beta)}(x)
  Q_{k,n}^{(\alpha,\beta)}(y)\sqrt{w_n^{(\alpha,\beta)}(x)w_n^{(\alpha,\beta)}(y)},
  \qquad x,y\in[0,n],
\end{equation}
that is the Christoffel--Darboux kernel for the Hahn polynomials
\eqref{eq:orthoQ}.

Let us now focus on the fluctuations of the portion $\Gamma_1$ of the
Arctic curve.  Our starting point is the expression for EFP, as given
in Prop.~\ref{prop:EFP}. Let us consider values of $p$ within the
interval $[N,p_{\mathrm{max}}]$, for some $p_{\mathrm{max}}<\lfloor
(M-L)(L-N)/(M-L+N)\rfloor$. The lower bound has been chosen for
simplicity, while the upper bound is such that, when setting
$p=\lfloor x \ell \rfloor$, in the scaling limit one has $x\leq
x_{\mathrm{max}}< \xc$, see \eqref{eq:xc}, and the considered portion
of $\Gamma_1$ is thus away from the contact point. The expression for
EFP simplifies to
\begin{equation}\label{eq:efp_simple}
  F_{L,M,N}(p,q)=H(M-p-q,p-N,L-N-p,N,M-L+N),
\end{equation}
with $p\in[N,p_{\mathrm{max}}]$. Let us fix a value of $p$, and denote
by $\xi$ the value of the topmost thick edge between the $p$th and
$(p+1)$th vertical lines of the $L\times M$ lattice.  It follows from
the definition of EFP that
\begin{equation}
\mathbb{P}(\xi< M-q)=F_{L,M,N}(p,q).
\end{equation}
Recalling now \eqref{eq:FD1} and \eqref{eq:efp_simple}, we 
 may write
\begin{equation}
  \mathbb{P}(\xi <M-q )=\det[1-K_{M-L+N,N}|_{(M-p-q,M-L+N]}],
\end{equation}
with parameters $\alpha$ and $\beta$ of the Hahn measure specialized
to $\alpha=p-N$ and $\beta=L-N-p$.

We may now consider values of $M-q$ in the vicinity of the portion
$\Gamma_1$ of the arctic curve, as given by Thm.~\ref{thm:result}.  In
the non-intersecting lattice path picture, see
Sect.~\ref{sec:bijection}, this amounts to consider values $M-q$ in
the vicinity of $p+\ell R_0$, with $R_0:=R(\alpha_0,\beta_0,s_0,n_0)$
given by \eqref{right-end-point}. It has been conjectured in
\cite{J-00}, and proven in \cite{BKMM-07}, see Thm.~3.14 therein,
that, in such regime, the Christoffel--Darboux kernel for the Hahn measure
tends to the Airy kernel, in the scaling limit. Specifically, we have,
for some constant $t>0$,
\begin{equation}
\lim_{\ell\to\infty}\mathbb{P}\left(\frac{\xi-p-\ell
  R(\alpha_0,\beta_0,s_0,n_0)}{(t\ell)^{1/3}}\leq x\right)=
\det[1-A|_{(x,\infty)}],
\end{equation}
where the integral operator $A|_{(x,\infty)}$ acts on $L^2[x,\infty)$
with the Airy kernel. In other words, the fluctuations of the
considered portion of $\Gamma_1$ are indeed governed by the
Tracy-Widom distribution. 

One may proceed similarly for the fluctuations of the portion
$\Gamma_2$ of the Arctic curve. However, on the one hand, the
expression for AFP is more intricate, and, on the other hand, before
applying the procedure carried out above for EFP, one should use
particle-hole duality (in the sense of discrete log-gases on finite
intervals, see Sect.~3.2 in \cite{BKMM-07}).  This makes the
discussion slightly more involved, but leads to the conclusion that
fluctuations of the portion $\Gamma_2$ of the Arctic curve, away from
the contact points, are again governed by the Tracy--Widom
distribution.

To conclude, we comment that Airy type density fluctuations in the
vicinity of $\Gamma_1$ and $\Gamma_4$ are indeed clearly visible in
Fig.~\ref{fig:numerics}, left, and similarly, for $\Gamma_3$ and
$\Gamma_6$, in Fig.~\ref{fig:numerics}, center. However the phenomenon
is not as apparent for $\Gamma_2$ nor $\Gamma_6$, in
Fig.~\ref{fig:numerics}, right. This is due to the fact that
oscillations in the densities of $c$-vertices of type 3 or 4, see
Fig.~\ref{fig:weights}, are in antiphase, and cancel out when
considering the total density of $c$ vertices.
To counter this problem, following \cite{KL-17}, we plot the
difference between the two densities associated to the two types of
$c$-vertices, see Fig.~\ref{fig:num_c1-c2}. Now Airy type density
fluctuations in the vicinity of $\Gamma_2$ and $\Gamma_5$ become
indeed clearly visible.

\begin{figure}[t]
\includegraphics[width=.4\textwidth]{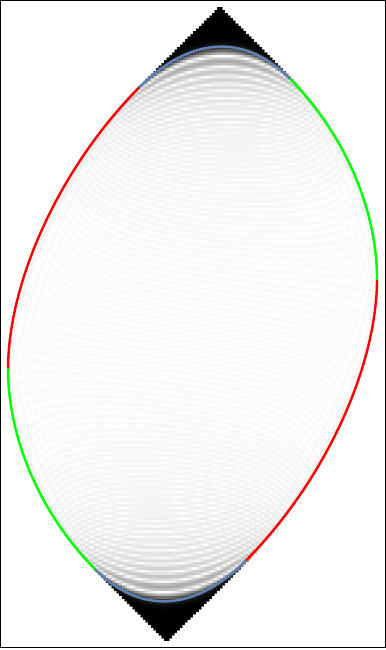}
\caption{Absolute value of the difference of densities of $c$-vertices
  (that is of vertices of the third and fourth type, in Fig.~1), with
  color varying from white to black as densities vary from 0 to 1. The
  picture has been generated by averaging over $10^5$ configurations
  of the lattice with $N=60$, $L=140$, $M=240$.}
\label{fig:num_c1-c2}
\end{figure}


\subsection*{Acknowledgments}
We are indebted to N.M. Bogoliubov, L. Cantini, I. Prause, D. Serban,
A. Sportiello, J.-M. St\'ephan, J. Viti, for stimulating
discussions at various stages of this work.  FC and AM are grateful to
the Workshop on `Randomness, Integrability, and Universality', held on
Spring 2022 at the Galileo Galilei Institute for Theoretical Physics,
for hospitality and support at some stage of this work. AGP is
grateful to INFN, Sezione di Firenze, for hospitality and partial
support during his stay in Florence, Italy, where a part of this work
was done.  FC acknowledges partial support from Italian Ministry of
Education, University and Research under the grant PRIN 2017E44HRF on
`Low-dimensional quantum systems: theory, experiments and
simulations'. AGP acknowledges support from the Russian Science
Foundation, grant 21-11-00141.

\appendix

\section{Boundary-refined partition  function}\label{app:boundary}

We want here to evaluate the boundary-refined partition function. We
shall resort to the bijection between configurations of the
four-vertex model and non-intersecting lattice paths introduced in
Sec.~\ref{sec:bijection}.

As a warm-up up, let us evaluate first the partition function
$Z_{L,M,N}$. This is the number of configurations of $N$ lattice paths
constrained by the four-vertex rules, with the the $j$th path,
$j=1,\dots,N$, counting from the left, connecting the two vertices of
coordinates $(j,1)$ and $(L-N+j,M)$, respectively, see
Fig.~\ref{fig:NILP}, left.  Such configurations are equinumerous with
those of $N$ non-intersecting lattice paths connecting the vertices
$(j,1)$ and $(L-N+j,K)$, with $K=M-L+N+1$, see Fig.~\ref{fig:NILP},
right. Using the Linstr\"om-Gessel-Viennot lemma \cite{L-73,GV-85}, we
thus have:
\begin{equation}\label{eq:Z=det}
Z_{L,M,N}=\det_{1\leq i,j \leq N}\left[\binom{M+i-j}{L-N+i-j}\right].
\end{equation}
Determinants of matrices of binomials as above may be evaluated
explicitly, see \cite{K-01}, Thm.~26, yielding \eqref{eq:Z=B}.

Let us turn now to the boundary-refined partition function
$Z_{L,M,N}^{(1)}(n)$.  Recall that in the four-vertex model
description, on the original rectangular lattice, we require the
presence of a vertex of type $b$ at position $(n,m)$, with
$m=n+M-L+N-1$, see Fig.~\ref{fig:boundary-4vm-NILP}, left.  This means
that the leftmost path, starting from vertex $(1,N)$, is constrained
to reach the vertex $(n,m-1)$. As for the remaining paths, the $j$th
one, $j=2,\dots,N$, must connect the vertices of coordinates
$(j,N-j+1)$ and $(L-N+j,M-1)$.

In terms of non-intersecting lattice paths, the above conditions
translate as follows: the $j$th path must connect the vertices of
coordinates $(j,N-j+1)$ and $(x_j,K-1)$,  where $K=M-L+N+1$, and
\begin{equation}\label{eq:xjH}
  x_j=
  \begin{cases}
    n, & \quad j=1,\\
    L-N+j, & \quad j=2,\dots,N,
  \end{cases}
\end{equation}
see Fig.~\ref{fig:boundary-4vm-NILP}, right.  Application of
Linstr\"om-Gessel-Viennot lemma directly yields
\begin{equation}\label{eq:Zref2} 
	Z_{L,M,N}^{(1)}(n) = \det_{1\leq i,j\leq N}
			\binom{K-N-2 + x_j }{x_j-i}.
\end{equation}
The determinant of binomials in \eqref{eq:Zref2} can be evaluated
explicitly:
\begin{equation}
	Z_{L,M,N}^{(1)}(n) =
	\prod_{1\le i < j \le N} (x_j-x_i)
	\prod_{j=1}^{N}
	\frac{ (K-2 + x_j - N)! }
			 { (x_j - 1)! (K-1-j)! }.
\end{equation}
Taking into account \eqref{eq:xjH}, we find that
\begin{multline}
Z_{L,M,N}^{(1)}	(n) = 
\frac{(M-L-1+n)! (L-n)!}{(n-1)!(L-N-n+1)!}
\\\times 
\frac{1}{(M-L+N-1)!}
\prod_{2\leq i<j\le N} (j-i) 
\prod_{j=2}^{N} \frac{(M-N-1+j)!}{ (L-N+j-1)!(M-L+N-j)! } .
\end{multline}
Simplifying further we easily  recover  \eqref{eq:Zrefined}.

\section{Column-refined partition function}\label{app:column-Z}

Let us introduce here a quantity that, although somewhat intermediate,
and formulated for the model of non-intersecting lattice paths,
appears useful to evaluate the probability of occurrence of a variety
of configurations in the four-vertex model.

Let us consider $N$ non-intersecting lattice paths on the $L\times K$
lattice, with our usual boundary conditions. In the spirit
of \cite{CDP-21}, consider the partition function of the model
when the paths are conditioned to flow through the $N$ horizontal
edges lying in the $n$th column (i.e., between the $n$th and $(n+1)$th
vertical lines), with ordinates $\mathbf{y}=\{y_1,\dots,y_N\}$, such
that $1\leq y_1<\dots<y_N\leq K$.  We denote such partition function
by $Z_{L,M,N}(n,\mathbf{y})$.

Note that we have reversed the labelling of paths, with respect to
App.~\ref{app:boundary}. Also, it is convenient to use the freedom
discussed in Sec.~\ref{sec:bijection} to slightly change the boundary
condition in such a way to have exactly $N$ horizontal thick edges in
each column, including the first $N$ ones, and the last $N$ ones. We
therefore choose as endpoints of the $j$th path, $j=1,\dots,N$,
counting from the right, the vertices of coordinates $(1,j)$ and
$(L,K-N+j)$, respectively, see Fig.~\ref{fig:MPCF-WF}.

\begin{proposition}\label{prop:Z_LMN(n,y)}
  The number of configurations of $N$ non-intersecting lattice paths,
  on the $L\times K$ lattice, with endpoints as mentioned above, and 
  conditioned to flow through the horizontal edges at positions
  $(n,y_1)$, \dots, $(n,y_N)$, is:
\begin{equation}\label{eq:prop_Zy1}
  Z_{L,M,N}(n,\mathbf{y})=C_{L,M,N;n}
  \prod_{j=1}^{\tilde{n}}\delta_{y_j,j}\prod_{j=\tilde{l}+1}^{N}\delta_{y_j,K-N+j}
  \prod_{\tilde{n}< i<j\leq\tilde{l}}(y_j-y_i)^2\prod_{j=\tilde{n}+1}^{\tilde{l}}\mu(y_j),
\end{equation}
where
\begin{equation}\label{eq:1p-measure}
  \mu(y)=\binom{y-N+n+\tilde{n}-1}{y-\tilde{n}-1}
  \binom{K+L-n-\tilde{l}-y}{K-N+\tilde{l}-y}
\end{equation}
and 
\begin{multline}
  C_{L,M,N;n}=
  \prod_{i=1}^{\tilde{n}}\frac{(i-1)!(K+L-N-n-i)!}{(K-N+\tilde{l}-i)!(L-n-i)!}
  \\
  \times \prod_{i=\tilde{n}+1}^{\tilde{l}}
  \frac{(2\tilde{n}-N+n)!(L+N-n-2\tilde{l})!}{(i-N+n-1)!(L-n-i)!}
   \\
  \times
  \prod_{i=\tilde{l}+1}^{N}
  \frac{(N-i)!(K-2N+n+i-1)!}{(K-N+\tilde{n}+i-1)!(i-N+n-1)!}
  \label{eq:prop_Zy2},
\end{multline}
with $\tilde{n}=\max(N-n,0)$, and
$\tilde{l}=\min(L-n,N)$.
\end{proposition}
\begin{proof}
In order to evaluate $Z_{L,M,N}(n,\mathbf{y})$, let us split the
lattice into two parts, by `cutting' all horizontal edges between the
$n$th and $(n+1)$th lines.  It is natural to introduce two `wave
functions' $\Psi_1(\mathbf{y})$ and $\Psi_2(\mathbf{y})$, that, for a
given configuration of thick horizontal edges, are the partition
functions of the west and east portions of the split lattice,
respectively, see Fig.~\ref{fig:MPCF-WF}.

\begin{figure}[t]
\includegraphics{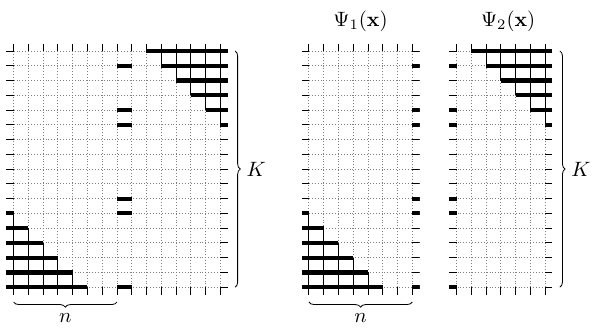}
\caption{The column-refined partition function
  $Z_{L,M,N}(n,\mathbf{y})$ (left), and the two wave functions,
  $\psi_1(\mathbf{y})$ and $\psi_2(\mathbf{y})$ (right).  Here $L=15$,
  $K=17$, $N=6$, $n=8$, and $\mathbf{y}=(1,6,7,12,13,16)$.}
\label{fig:MPCF-WF}
\end{figure}

The two wave functions are easily calculated. Indeed
$\Psi_1(\mathbf{y})$ is the number of configurations of $N$
non-intersecting lattice paths connecting the vertices of coordinates
$(1,j)$ and $(n,y_j)$, $j=1,\dots,N$. Similarly for
$\Psi_2(\mathbf{y})$, but now with the vertices of coordinates
$(n+1,y_j)$ and $(L,K-N+j)$, $j=1,\dots,N$. Application of the
Lindstr\"om-Gessel-Viennot lemma yields
\begin{align}
  \Psi_1(\mathbf{y})&=\det_{1\leq i,j\leq N}
  \left[\binom{y_i-j+n-1}{y_i-j}\right],
  \\
  \Psi_2(\mathbf{y})&=\det_{1\leq i,j\leq N}
  \left[\binom{K-N+L-n-y_i+j-1}{K-N-y_i+j}\right].
\end{align}
These determinants evaluate to:
\begin{align}\label{eq:1p-psi1}
  \Psi_1(\mathbf{y})&=\prod_{1\leq i <j\leq N}(y_j-y_i)
  \prod_{i=1}^N\frac{(y_i-N+n-1)!}{(y_i-1)!(i-N+n-1)!},
  \\
  \Psi_2(\mathbf{y})&=\prod_{1\leq i <j\leq N}(y_j-y_i)
  \prod_{i=1}^N\frac{(K+L-N-n-y_i)!}{(K-y_i)!(L-n-i)!},\label{eq:1p-psi2}
\end{align}
where, once more, we have used \cite{K-01}, Thm.~26.

Note that, strictly speaking, the expression \eqref{eq:1p-psi1} for
$\Psi_1(\mathbf{y})$ is well defined only when $n\geq N$, being
otherwise plagued with negative factorials. This corresponds to the
fact that, when $n<N$, the $N-n$ lowest thick edges are constrained to
$y_j=j$, $j=1,\dots,N-n$.  Implementing such constraint, the negative
factorials indeed cancel out, as it may be easily seen, e.g.,
interpreting them in terms of some suitable limit of corresponding
Gamma functions. We obtain
\begin{equation}\label{eq:1p-psi21}
  \Psi_1(\mathbf{y})=\prod_{i=1}^{\tilde{n}}\delta_{y_i,i}
  \prod_{\tilde{n}+1\leq i<j\leq N}(y_j-y_i)
  \prod_{i=\tilde{n}+1}^N\frac{(y_i-N+n-1)!}{(y_i-\tilde{n}-1)!(i-N+n-1)!},
\end{equation}
where $\tilde{n}:=\max(N-n,0)$, and we follow the usual
convention that empty products, such as the first one when $n>N$,
equal $1$. Representation \eqref{eq:1p-psi21} for the wave function
$\Psi_1(\mathbf{y})$ holds for $1\leq n\leq L$.

Similarly, the expression \eqref{eq:1p-psi2} for $\Psi_2(\mathbf{y})$,
becomes ill-defined for $n>L-N$, unless the highest thick edges are
constrained to $y_j=K-N+j$, $j=L-n+1,\dots,N$. Proceeding as above, we
get
\begin{equation}\label{eq:1p-psi22}
   \Psi_2(\mathbf{y})=\prod_{1\leq i<j\leq\tilde{l}}(y_j-y_i)
  \prod_{i=1}^{\tilde{l}}\frac{(K+L-N-n-y_i)!}{(K-N+\tilde{l}-y_i)!(L-N-i)!}
  \prod_{i=\tilde{l}+1}^N\delta_{y_i,K-N+i},
\end{equation}
where $\tilde{l}:=\min(L-n,N)$. Representation \eqref{eq:1p-psi22} for
the wave function $\Psi_2(\mathbf{y})$ holds for $1\leq n\leq
L$. Recalling that the column-refined partition function
$Z_{L,M,N}(n,\mathbf{y})$ is simply the product of
$\Psi_1(\mathbf{y})$ and $\Psi_2(\mathbf{y})$, and implementing the
Kronecker deltas, setting $y_i=i$, $i=1,\dots,\tilde{n}$ in
$\Psi_2(\mathbf{y})$, and $y_i=K-N+i$, $i=\tilde{l}+1,\dots,N$ in
$\Psi_1(\mathbf{y})$, we recover \eqref{eq:prop_Zy1}.
\end{proof}

\begin{remark}
If we restrict $n$ to the range $N\leq n\leq L-N$, expression
\eqref{eq:prop_Zy1} simplifies significantly:
\begin{multline}
  Z_{L,M,N}(n,\mathbf{y})=C_{L,N;n} \prod_{1\leq i <j\leq
    N}(y_j-y_i)^2
  \\
  \times \prod_{i=1}^N \binom{y_i-N+n-1}{y_i-1}
  \binom{K+L-N-n-y_i}{K-y_i},
\end{multline}
with
\begin{equation}
  C_{L,N;n}=
  \prod_{i=1}^N\frac{(n-N)!(L-N-n)!}{(n-i)!(L-n-i)!}.
\end{equation}
\end{remark}
 
Prop.~\ref{prop:Z_LMN(n,y)} allows for evaluation of various
quantities of interest for the study of the four-vertex model.

\section{Proof of Propositions \ref{prop:EFP}, \ref{prop:EFP_AFP},
  and \ref{prop:AFP}}\label{app:EFP_AFP}

Let us proceed with the proof of Prop.~\ref{prop:EFP}.
\begin{proof}
To start with, we observe that the condition entering the definition
of EFP, namely that all vertices in $\Theta$ should be of type $a$,
may be rephrased in the rectangular domain as the requirement that
there is no path flowing through the topmost $q$ consecutive
horizontal edges, in the $p$th column. This can in turn be
equivalently rephrased in terms of non-intersecting lattice paths
along the lines of Sect.~\ref{sec:bijection}, as follows. On the
$L\times K$ lattice, there is no path flowing through the topmost
$\tilde{q}$ horizontal edges, in the $p$th column, with
$\tilde{q}:=p+q-L+N$, see Fig.~\ref{fig:EFPandAFP-NILP-constraints}.
We may therefore write
\begin{equation}
   	F_{L,M,N}(p,q)=\frac{1}{Z_{L,M,N}}\sum_{1\leq \mathbf{y}\leq
          K-\tilde{q}} Z_{L,M,N}(n,\mathbf{y}),
\end{equation}
where the column-refined partition function $Z_{L,M,N}(n,\mathbf{y})$ has
been defined in App.~\ref{app:column-Z}. 

\begin{figure}[t]
\includegraphics{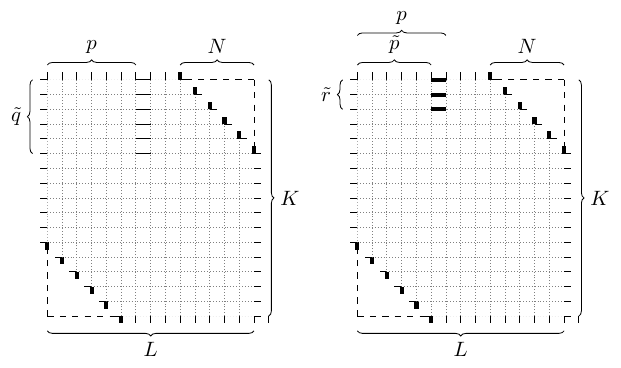}
\caption{Constraints corresponding to the EFP (left) and AFP (right)
  in terms of non-intersecting lattice path. Here $L=15$, $M=25$,
  $N=6$, $p=7$, $q=8$, and therefore $K=17$, $\tilde{q}=6$, $\tilde{p}=6$,
  $\tilde{r}=3$.}
\label{fig:EFPandAFP-NILP-constraints}
\end{figure}

We now resort to Prop.~\ref{prop:Z_LMN(n,y)}. With no loss of
generality, we may restrict to $p\leq L-N$, see \eqref{eq:EFP-cond2},
that is, $\tilde{l}=N$, and set $\tilde{n}=\max(N-p,0)$. We get
\begin{equation}
  F_{L,M,N}(p,q)=\frac{C_{L,M,N;p}}{Z_{L,M,N}}\sum_{\tilde{n}<\mathbf{y}\leq
    K-\tilde{q}} \,\prod_{\tilde{n}<i<j\leq
    N}(y_j-y_i)^2\prod_{i=\tilde{n}+1}^N \tilde{\mu}(y_i),
\end{equation}
where $\tilde{\mu}(y)$ is given by $\mu(y)$ in \eqref{eq:1p-measure},
but with $n$ and $\tilde{l}$ replaced by $p$ and $N$, respectively.
Relabelling and shifting the variables, $y_i\to
x_{i+\tilde{n}}+\tilde{n}+1$, and observing that
\begin{equation}
  \tilde{\mu}(x+\tilde{n}+1)=w_{K-\tilde{n}-1}^{(p-N+2\tilde{n},L-N-p)}(x),
  \end{equation}
we finally get 
\begin{align}
  &F_{L,M,N}(p,q)=\\
  &=\frac{C_{L,M,N;p}}{Z_{L,M,N}}\sum_{0\leq
    \mathbf{y}\leq K-\tilde{q}-\tilde{n}-1} \,\prod_{1\leq i<j\leq
    N-\tilde{n}}(y_j-y_i)^2 \prod_{i=1}^{N-\tilde{n}}
  w_{K-\tilde{n}-1}^{(p-N+2\tilde{n},L-N-p)}(x_i)
  \\ &=H(K-\tilde{n}-\tilde{q}-1,p-N+2\tilde{n},
  L-N-p,N-\tilde{n},K-\tilde{n}-1),
\end{align}
where the function $H(\cdot)$ has been defined in
\eqref{eq:H=def}. Recalling that $\tilde{q}=p+q-L+N$, $K=M-L+N-1$, we
recover Prop.~\ref{prop:EFP}.
\end{proof}

Let us now turn to the proof of Prop.~\ref{prop:EFP_AFP}.

\begin{proof}
We want to show that AFP may be expressed in terms of EFP, via a
suitable identification of the parameters.  We shall resort once again
to the bijection between configurations of the four-vertex model and
non-intersecting lattice paths.

As already observed, the definition of EFP, namely that, in the
hexagonal domain, the topmost $\tilde{q}=p+q-L+N$ vertices in the
$p$th vertical line are all of type $a$, once rephrased in terms of
non-intersecting lattice paths, reads: on the $L\times K$ lattice, no
path flows through the topmost $\tilde{q}$ horizontal edges, in the
$p$th column.  Similarly, the definition of AFP implies the flow of a
path through each of the $\tilde{r}$ topmost horizontal edges, in the
$\tilde{p}$th column, where
\begin{equation}
  \tilde{p} =p -\mathrm{mod}(\tilde{q}+1,2),\qquad
  \tilde{r}=\bigg\lfloor\frac{\tilde{q}+1}{2}\bigg\rfloor.
\end{equation}
Indeed, the requirement that, in the hexagonal domain, the topmost
$\tilde{q}$ vertices in the $p$th vertical line are all of type $c$,
implies having a path flowing through each of the topmost $\tilde{r}$
horizontal edges to the right (and left) of the $p$th line when
$\tilde{q}$ is odd (even), see
Fig.~\ref{fig:EFPandAFP-NILP-constraints}.

\begin{figure}[t]
\includegraphics[width=.9\textwidth]{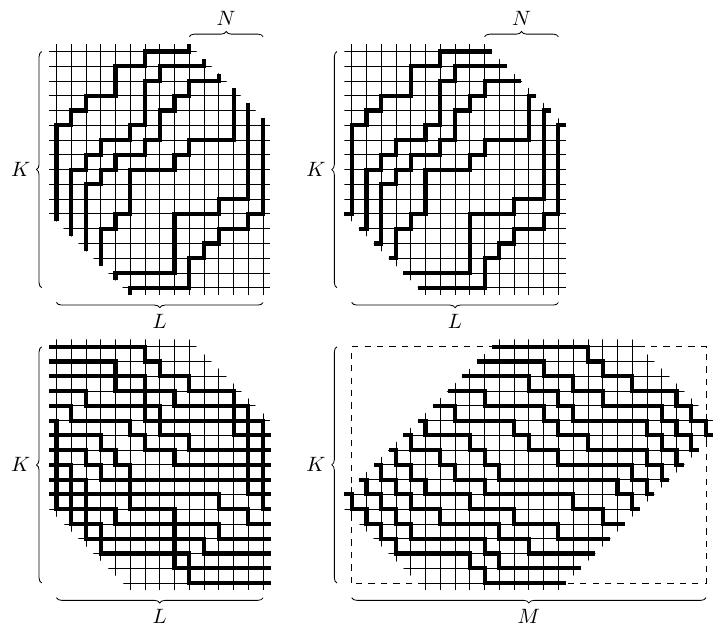}
\caption{An example of the procedure used to relate EFP and AFP.
  Starting from a give configuration of non-intersecting lattice
  paths, with frozen corners removed (top left), swap the states of
  all edges in the frozen triangular corners (top right), swap the
  state of all horizontal edges (bottom left), and shift all paths
  eastward, obtaining a four-vertex model configuration on a rotated
  $K\times M$ lattice with $K-N$ lines (bottom right). Here $L=15$,
  $K=17$, $N=6$, and hence $M=25$.}\label{fig:AFP-bijection}
\end{figure}

In order to express AFP in terms of EFP, we consider another bijection
between configurations of the four-vertex model and non-intersecting
lattice paths. For any given configuration of $N$ non-intersecting
paths on the $L\times K$ rectangular domain, let us implement the
following procedure \textit{i)} swap the state (full/empty) of each
edge in the northeast and southwest frozen triangular corners; in
other words, in these two corners, paths are transformed from all
vertical to all horizontal; this action does not change the number of
allowed configurations; \textit{ii)} swap the state of all horizontal
edges; as a result, we have $K-N$ lattice paths; these may now
osculate, but we assume they do not intersect, so that each path is
still uniquely defined; \textit{iii)} shift the $i$th path,
$i=1,\dots,K-N$, enumerated from the top, by $K-N-i$ steps to the
right. Note that, with the last step, the horizontal size of domain
becomes $L+K-N-1=M$. An example of this procedure is given in
Fig.~\ref{fig:AFP-bijection}.

All steps of the above procedure being invertible, it follows that
each configuration of $N$ non-intersecting paths on the $L\times K$
rectangular domain is bijectively mapped into a configuration of the
four-vertex model on the $K\times M$ domain (rotated by $\pi/2$), with
$K-N$ lines.

\begin{figure}[t]
\includegraphics[width=\linewidth]{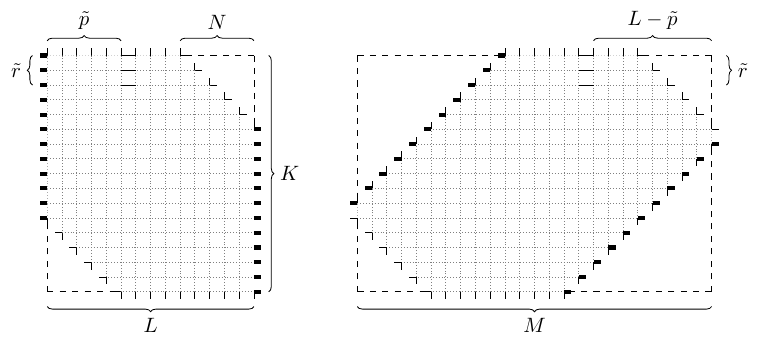}
\caption{The AFP-constraint with inversed horizontal edges (left) and
  equivalent EFP-constraint (right).}\label{fig:AFP-bijection-BC}
\end{figure}

Note that under this bijection the condition defining the AFP, namely
the flow of a path through each of the $\tilde{r}$ topmost horizontal
edges, in the $\tilde{p}$th column, becomes the condition of having
only vertices of type $a$ within a rectangular region of size
$\tilde{r}\times (L-\tilde{p})$ in the northeast corner of the
$K\times M$ domain with $N-K$ lines, see
Fig.~\ref{fig:AFP-bijection-BC}. In formulae:
\begin{equation}
  G_{L,M,N}(p,q)=F_{K,M,K-N}(\tilde{r},L-\tilde{p}).
\end{equation}
Restoring the dependence from the original parameters, we recover
Prop.~\ref{prop:EFP_AFP}.
\end{proof}

Finally, let us derive Prop.~\ref{prop:AFP}.

\begin{proof}
 It follows from Prop.~\ref{prop:EFP} and Prop.~\ref{prop:EFP_AFP} that
  AFP may be written according to \eqref{eq:G(p,q)=}, with parameters
\begin{gather}
	d = \min(\tilde{r},M-L+1)-\tilde{r}+\tilde{p}-1,  \qquad 
	\alpha = \abs{M-L-\tilde{r}+1}, \qquad
	\beta = N-\tilde{r}, \\
	s =  \min(\tilde{r},M-L+1),\qquad
	n = L-N+\min(\tilde{r},M-L+1)-1.
\end{gather}
Noting that $2\tilde{r}-\tilde{p}=\tilde{q}-p+1$, implying $
L-N+\tilde{r}-\tilde{p}=q+1-\tilde{r}$, we recover expression
\eqref{eq:G(p,q)parameters=} for the parameters.
\end{proof}

\bibliography{4vmAC_bib.bib}

\end{document}